\documentclass[10pt,journal,compsoc]{IEEEtran}
\usepackage{graphicx}
\usepackage[cmex10]{amsmath}
\usepackage{amssymb}
\usepackage{amsthm}
\usepackage{array}
\usepackage{url}
\usepackage{algorithm,algorithmic}
\usepackage{chemarrow}
\usepackage{multirow}
\usepackage{extarrows}
\usepackage{setspace}
\usepackage{booktabs}
\usepackage{tabulary}
\usepackage[shortlabels]{enumitem}
\usepackage{balance}
\usepackage{subfigure}
\usepackage[cmex10]{amsmath}
\usepackage[shortlabels]{enumitem}
\usepackage{amssymb}
\usepackage{amsthm}
\usepackage{multirow}
\usepackage{cite}
\usepackage{color}
\usepackage{setspace}
\usepackage{ragged2e}
\usepackage{stmaryrd} 
\usepackage{tabularx}
\usepackage{makecell}
\usepackage{footnote}
\usepackage{threeparttable}
\usepackage{eqparbox}
\usepackage{bm}
\usepackage[colorlinks,
linkcolor=black,
anchorcolor=black,
citecolor=black
]{hyperref}

\makeatletter
\newcommand{\ALOOP}[1]{\ALC@it\algorithmicloop\ #1%
  \begin{ALC@loop}}
\newcommand{\ENDALOOP}{\end{ALC@loop}\ALC@it\algorithmicendloop}

\newtheorem{theorem}{\textbf{\emph{Theorem}}}
\newtheorem{definition}{\textbf{\emph{Definition}}}
\newcommand{\graph}{graph-structured data}
\newcommand{\secgnn}{SecGNN}
\newcommand{\main}{SecGNN}
\floatname{algorithm}{Subroutine}

\begin{document}

\title{SecGNN: Privacy-Preserving Graph Neural Network Training and Inference as a Cloud Service}

\author{Songlei~Wang, Yifeng~Zheng, and Xiaohua~Jia,~\IEEEmembership{Fellow,~IEEE}
	
	\IEEEcompsocitemizethanks{
		\IEEEcompsocthanksitem Songlei Wang and Yifeng Zheng are with the School of Computer Science and Technology, Harbin Institute of Technology, Shenzhen, Guangdong 518055, China (e-mail: songlei.wang@outlook.com; yifeng.zheng@hit.edu.cn).
		\IEEEcompsocthanksitem Xiaohua Jia is with the School of Computer Science and Technology, Harbin Institute of Technology, Shenzhen, Guangdong 518055, China, and also with the Department of Computer Science, City University of Hong Kong, Hong Kong, China (e-mail: csjia@cityu.edu.hk)
		
	    \IEEEcompsocthanksitem Corresponding author: Yifeng Zheng.
		
	}
}

\IEEEtitleabstractindextext{
\begin{abstract}
Graphs are widely used to model the complex relationships among entities. As a powerful tool for graph analytics, graph neural networks (GNNs) have recently gained wide attention due to its end-to-end processing capabilities. With the proliferation of cloud computing, it is increasingly popular to deploy the services of complex and resource-intensive model training and inference in the cloud due to its prominent benefits. However, GNN training and inference services, if deployed in the cloud, will raise critical privacy concerns about the information-rich and proprietary graph data (and the resulting model). While there has been some work on secure neural network training and inference, they all focus on convolutional neural networks handling images and text rather than complex graph data with rich structural information. In this paper, we design, implement, and evaluate SecGNN, the first system supporting privacy-preserving GNN training and inference services in the cloud. SecGNN is built from a synergy of insights on lightweight cryptography and machine learning techniques. We deeply examine the procedure of GNN training and inference, and devise a series of corresponding secure customized protocols to support the holistic computation. Extensive experiments demonstrate that SecGNN achieves comparable plaintext training and inference accuracy, with promising performance.
\end{abstract}

}

 \maketitle

\IEEEdisplaynontitleabstractindextext

\IEEEpeerreviewmaketitle

\section{Introduction}

\IEEEPARstart{G}raphs have been widely used to model and manage data in various real-world applications, including recommendation systems \cite{zhang2019deep}, social networks \cite{kim2018social} and webpage networks \cite{scarselli2008graph}.
Graph data, however, is highly complex and inherently sparse, making graph analytics challenging \cite{wu2021learning}.
With the rapid advancements in deep learning, Graph Neural Networks (GNNs) \cite{wu2020comprehensive} have recently gained a lot of traction as a powerful tool for graph analytics due to its end-to-end processing capabilities.
GNNs can empower a variety of graph-centric applications such as node classification \cite{ying2018graph}, edge classification \cite{kim2019edge} and link prediction \cite{zhang2018link}.
With the widespread adoption of cloud computing, it is increasingly popular to deploy machine learning training and inference services in the cloud \cite{aws,azure}, due to the well-understood benefits \cite{QinW0018,JiangWHWLSR21}.
However, GNN training and inference, if deployed in the public cloud, will raise critical severe privacy concerns.
Graph data is information-rich and can reveal a considerable amount of sensitive information.
For example, in a social network graph, the connections between nodes represent users' circles of friends and each node's features represent each user's preferences. 
Meanwhile, the graph data as well as the trained GNN model are the proprietary to the data owner, so revealing them may easily harm the business model.
Therefore, security must be embedded in outsourcing GNN training and inference to the cloud.

In the literature, privacy-preserving machine learning has received great attention in recent years, especially the design of secure protocols for neural network-based applications. A number of research efforts have been proposed for secure neural network inference and training.
Most of existing works \cite{gilad2016cryptonets,liu2017oblivious,juvekar2018gazelle,riazi2019xonn,chaudhari2019astra,patra2020blaze,mishra2020delphi,kumar2020cryptflow,rathee2020cryptflow2} are focused on designing specialized protocols for secure inference, and only a few works \cite{mohassel2017secureml,mohassel2018aby3, wagh2019securenn,wagh2020falcon, tan2021cryptgpu} study secure training which is more sophisticated and resource-intensive.
However, prior works are all focused on the support for Convolutional Neural Networks (CNNs) handling unstructured data like images and text.
How to achieve secure in-the-cloud training and inference of GNNs that handle complex graph data remains unexplored.

Supporting secure training and inference of GNNs in the cloud, however, faces unique challenges and require delicate treatments due to the complex structured nature of graphs.
There are various kinds of structural information in graphs: 1) relationships between nodes (i.e., edges), 2) edge weights, and 3) number of neighboring nodes (i.e., degrees of nodes).
Designing solutions for securing GNN training and inference thus demands protection for not only numerical information (e.g., the values of features associated with nodes) but also the rich structural information unique to different graphs.

In light of the above, in this paper, we present the first research endeavor towards privacy-preserving training and inference of GNNs in the cloud.
We design, implement, and evaluate a new system SecGNN, which allows a data owner to send encrypted graph data to the cloud, which can then effectively train a GNN model without seeing the graph data as well as provide secure inference once an encrypted GNN model is trained.
Targeting privacy assurance as well as high efficiency, SecGNN builds on only lightweight cryptographic techniques (mainly additive secret sharing) for efficient graph data encryption at the data owner as well as secure training and inference at the cloud side.
To be compatible with the working paradigm of additive secret sharing, SecGNN employs a multi-server and decentralized trust setting where the power of the cloud is split into three cloud servers that are hosted by independent cloud service providers.
The adoption of such a multi-server model to facilitate security applications in various contexts has gained increasing traction in prior works \cite{wagh2019securenn,MohasselRR20,tan2021cryptgpu} as well as in industry \cite{Mozilla,WhitePaper}. 
SecGNN leverages the above trend and contributes a new design point of secure GNN training and inference in the cloud through highly customized cryptographic protocols.

We start with considering how to appropriately encrypt the graph data in SecGNN so that it can still be effectively used at the cloud for secure training and inference.
As mentioned above, graphs contain not only numerical information (i.e., feature vectors associated with the nodes) but also structural information connecting the nodes, all demanding strong protection.
The challenge here is to how to encrypt the structural information in an effective and efficient manner.
One may try to directly encrypt the adjacency matrix of the graph of size $O(N^2)$, where $N$ is the number of nodes, with additive secret sharing.

Such a simple method, however, is neither efficient nor necessary.
Firstly, there can be tens of thousands or even millions of nodes in a graph for practical applications, leading to the adjacency matrix being of very large size.
Directly encrypting the adjacency matrix would incur significant overheads.
Secondly, graphs are usually sparse, leading to the adjacency matrix being sparse and filled with many zeros.
Encrypting all the zeros in the adjacency matrix would result in unnecessary cost as well.
To tackle this challenge, our insight is to devise a set of customized data structures to appropriately store and represent the structural information, and so the encryption is performed over these data structures rather than the original (big) adjacency matrix.
With our customized data structures, the complexity of encryption significantly reduces to $O(N\cdot d_{max})$, where $d_{max}$ is the maximum degree of nodes in the graph, and far less than the number of nodes (e.g., only $0.87\%$ to $6.2\%$ as practically observed in our experiments over several popular real-world graph datasets).

Subsequently, we consider how to securely perform training and inference at the cloud over the delicately-encrypted graph data in SecGNN.
Through an in-depth examination on the computation required in GNN training and inference, we decompose the holistic computation into a series of functions and devise corresponding tailored secure constructions with the lightweight additive secret sharing technique.
Specifically, we manage to decompose the whole procedure into secure feature normalization, secure neighboring states aggregation, secure activation functions, and secure model convergence evaluation.

SecGNN supports secure feature normalization through realizing secure division with effective approximation mechanisms.
For secure neighboring states aggregation in SecGNN, our insight is to transform the problem into secure array access over encrypted arrays and indexes.
We design a secure array access protocol building on the state-of-the-art yet achieving much improved efficiency, through customized mechanisms.
To support the secure evaluation of activation functions (ReLU and Softmax), SecGNN mainly leverages insights from digital circuit design and provides tailored protocols in the secret sharing domain, rather than relying on expensive garbled circuits as in prior work \cite{mohassel2017secureml}.

Last but not least, SecGNN provides the first mechanism for secure convergence evaluation, allowing fine-grained control on the secure training process.
This is in substantial contrast to prior work on secure (CNN) training which simply sets a fixed number for the training epochs and thus may not necessarily meet convergence.
The synergy of these customized secure and efficient components leads to SecGNN, the first system supporting secure GNN training and inference in the cloud.
The security of SecGNN is formally analyzed.
We implement SecGNN and conduct extensive experiments over multiple real-world graph datasets.
The evaluation results demonstrate that SecGNN, while providing privacy protection in training and inference, achieves comparable plaintext accuracy, with promising performance.

We highlight our contributions below:
\begin{itemize}
	
	\item We present  {\secgnn}, the \textit{first} system supporting privacy-preserving GNN training and inference as a cloud service, through a delicate synergy of lightweight cryptography and machine learning.
	
	\item We devise customized data structures to facilitate efficient and effective graph data encryption, and thoroughly propose a series of customized secure protocols to support the essential components required by secure GNN training and inference. 
	
	\item Among others, notably SecGNN provides a secure array access protocol with much improved efficiency over the state-of-the-art as well as the first secure fine-grained convergence evaluation protocol, which can be of independent interests.

	\item We make a full-fledged implementation of SecGNN and conduct an extensive evaluation over a variety of real-world graph datasets. The experiment results demonstrate the performance efficiency of SecGNN. 
	
\end{itemize}

The rest of this paper is organized as follows.
Section \ref{sec:related-work} discusses the related work.
Section \ref{sec:pre} introduces preliminaries.
Section \ref{sec:problem-statement} presents the problem statement. 
Section \ref{sec:main} gives the design of SecGNN. 
The security analysis is presented in Section \ref{sec:security-analysis}, followed by the experiments in Section \ref{sec:experiments}. Finally, we conclude this paper in Section \ref{sec:conclusion}.

\vspace{-10pt}
\section{Related Work}
\label{sec:related-work}

\subsection{Graph Neural Networks in Plaintext Domain}

Graphs can characterize the complex inter-dependency among data and are widely used in many applications, such as citation networks, social media networks, webpage networks \cite{wu2021learning}.
GNN models have the strong ability of capturing the dependence of graphs through message passing between the nodes of graphs, and have shown impressive performance in graph processing tasks.
The first GNN model was proposed in the seminal work of Scarselli et al. \cite{scarselli2008graph}.
Since then, many advanced GNN models targeting different applications and with varying capabilities have been put forward.
In general, GNN models can be divided into three categories: Gated Graph Neural Networks (GGNN) \cite{GGNN}, Graph Convolutional Networks(GCN) \cite{GCN}, and Graph ATtention networks (GAT) \cite{GAT}. 
GGNN models are proposed to accommodate applications that require to output sequences about a graph such as drug discovery \cite{chen2018rise}. 
GCN models are variants of CNNs which operate directly on graphs, and it is typically used for graphs with relatively stable nodes such as recommendation systems \cite{ying2018graph}.
GAT models introduce an attention-based architecture to calculate the weight of neighboring nodes, so that the whole network information can be obtained without knowing the structure of the whole graph, which is also commonly used in recommendation systems \cite{wang2019exploring}.
Although the above GNN models can achieve excellent performance on {\graph}, they are trained and work in the plaintext domain without considering privacy protection.
\vspace{-10pt}
\subsection{Secure Neural Network Training and Inference}

There has been a surge of interests on developing methods for secure neural network training and inference in recent years.
Most of existing works \cite{gilad2016cryptonets,liu2017oblivious,juvekar2018gazelle,riazi2019xonn,chaudhari2019astra,patra2020blaze,mishra2020delphi,kumar2020cryptflow,rathee2020cryptflow2} are focused on secure inference, and operated under different settings.
Some works \cite{gilad2016cryptonets,liu2017oblivious,juvekar2018gazelle,riazi2019xonn,mishra2020delphi,rathee2020cryptflow2} consider a 2-party setting where a model owner and a client directly engage in tailored cryptographic protocols for secure inference. Their security goal is that through the interactions the model owner learns no information while the client only learns the inference result.
In contrast, some works \cite{chaudhari2019astra,patra2020blaze,kumar2020cryptflow} consider an outsourced setting where a set of cloud servers are employed to perform secure inference over encrypted neural networks and inputs. Throughout the procedure, the cloud servers learn no information about the models, inputs, and inference results.
The cryptographic techniques adopted by the above works in different settings usually include homomorphic encryption, garbled circuits, and secret sharing.
In comparison with secret sharing, homomorphic encryption and garbled circuits are relatively expensive and usually incur large performance overheads.

In contrast with secure inference, secure training of neural networks is much more challenging because more complex operations would be required, and a large dataset needs to be processed in the ciphertext domain.
In the literature, only a few works study the problem of secure neural network training.
Mohassel et al. \cite{mohassel2017secureml} propose the first secure training method for shallow neural networks under a two-server setting, based on secret sharing (for linear operations) and garbled circuits (for approximated activation functions).
Subsequently, several works \cite{mohassel2018aby3, wagh2019securenn,wagh2020falcon, tan2021cryptgpu} achieve better performance in accuracy and efficiency by devising customized secure training protocols in a three-server setting.
Despite being useful, existing works on secure neural network are focused on CNN models that do not support the processing of graph data.
In light of this gap, in this paper we present the first research endeavor towards privacy-preserving training and inference of GNNs outsourced to the cloud, providing techniques for adequately encrypting graph data and securely supporting the essential operations required in GNN training and inference.
Following the trend as in prior work, our design adopts a similar three-server architecture, and only make use of lightweight cryptographic techniques in devising our secure protocol highly customized for secure GNN training and inference.

\subsection{Federated Learning-Based Private GNN Training}

There are some works \cite{meng2021cross,wu2021fedgnn,chen2021fedgraph,chen2022vertically,jiang2022federated,pei2021decentralized} focusing on privacy-preserving training of GNNs under the federated learning paradigm, where GNNs are trained across multiple clients holding local graph datasets in such a way that the graph datasets stay local.
Specifically, the work \cite{meng2021cross} focuses on GNNs over decentralized spatio-temporal data, and has the clients exchange model updates with the central server in plaintext.
In contrast, the works \cite{wu2021fedgnn,chen2021fedgraph} focus on distributed graph datasets where each client only holds a subgraph and design privacy-preserving mechanisms to protect the individual model updates.
Different from \cite{wu2021fedgnn,chen2021fedgraph}, the work \cite{chen2022vertically} focuses on vertically federated GNN, where all clients hold the same graph nodes, but different node features and edges.
The work \cite{jiang2022federated} considers federated dynamic GNN, which learns the representations of the objects at each timestamp by capturing the structural and patterns in the dynamic graph sequence.
Pei \textit{et al.} \cite{pei2021decentralized} focus on decentralized federated GNN, which allows multiple clients to train a GNN model without a centralized server and introduces the Diffie-Hellman key exchange method \cite{bonawitz2017practical} to achieve secure model aggregation between clients.
These federated learning-based works all target system models that are substantially different from ours.
{\main} targets an outsourced setting where the graph data owner can send its encrypted graph data to the cloud for secure training and can simply offline offline during the training process. In the meantime, {\main} readily supports secure GNN inference over encrypted GNNs and inputs as well, while those works can only deal with private training.

\subsection{Other Related Work}

There are some other works \cite{sajadmanesh2021locally,miao2022p} focusing on making the node features and edges differentially private \cite{dwork2006differential} when clients share their graph data to the central server or other clients during GNN models training.
Specifically, the work \cite{sajadmanesh2021locally} considers that a server holds a graph, whose nodes, which correspond to real users, have some private features that the server wishes to utilize for training a GNN model on the graph.
The work \cite{miao2022p} considers a distributed scenario where each client has all nodes but only partial private edges for a graph, and the clients wish to collaboratively train a GNN model on the distributed graph.
These works \cite{sajadmanesh2021locally,miao2022p} protect graph data privacy at the cost of notable accuracy degradation and rely on delicate parameter tuning for balancing accuracy and privacy.
In independent work, Wang \textit{et al.} \cite{wang2021privacy} propose a privacy-preserving representation learning framework on graphs from the mutual information perspective.
The framework considers a \emph{centralized} GNN training scenario and focuses on preventing the trained GNN models from leaking the training data by bounding the node features, node label, and link status during training GNN models.

\section{Preliminaries}
\label{sec:pre}

\subsection{Graph Neural Networks}
\label{sec:plainGNN}

A graph $\mathcal{G}=(\mathcal{V},\mathcal{E})$ consists of \textit{nodes} $\mathcal{V}$ and connections between nodes, i.e., \textit{edges} $\mathcal{E}$. Two nodes connected by an edge are \textit{neighboring nodes}.
The neighboring nodes of each node $v_{i}\in \mathcal{V}$ is denoted by $\{ne_{i,j}\}_{j\in{1,d_{i}}}$, where $d_{i}$ is node $v_{i}$'s \textit{degree}.
GNNs deal with graph-structured data, where each node in the graph is associated with a \textit{feature} vector and some of the nodes are \textit{labeled nodes}, each of which carries a \textit{classification label}.
Formally, we define the {\graph} in GNNs as $\mathcal{D}=\{\mathbf{A}, \mathbf{F},\mathbf{T}\}$.
Here, $\mathbf{A}$ is the \textit{adjacency matrix} of the graph, where $\mathbf{A}_{i,j}$ is an element in $\mathbf{A}$: If there exists an edge between node $v_{i}$ and node $v_{j}$, then $\mathbf{A}_{i,j}=1$ (binary graph) or $\mathbf{A}_{i,j}=w_{i,j}$ (weighted graph), and otherwise $\mathbf{A}_{i,j}=0$. In addition, Each row of $\mathbf{F}$  (denoted as $\mathbf{F}_{v_{i}}$) is node $v_{i}$'s feature vector, and each row of $\mathbf{T}$ (denoted as $\mathbf{T}_{v_{i}}$) is the classification label vector (one-hot encoding \cite{rodriguez2018beyond}) of labeled node $v_{i}\in\mathcal{T}$, where $\mathcal{T}$ is the set of labeled nodes.

Utilizing the {\graph} $\mathcal{D}$, a GNN model can be trained to perform graph analytic tasks.
In this paper, we focus on GCN as the first instantiation, which is well-established and the most representative GNN model \cite{GCN}.
With a trained GCN model, the classification labels of the unlabeled nodes can be inferred.
At a high level, this proceeds as follows.
Given an unlabeled node $v_{i}$, the trained GCN model infers its \textit{state} vector $\mathbf{x}^{(k)}_{v_{i}}$ (row vector) in the $k_{th}$-layer of the GCN. 
The dimension of the state vector decreases along with the layer propagation in the GCN.
The last layer state vector $\mathbf{x}^{(K)}_{v_{i}}$ is the inference result of node $v_{i}$, which is usually a probability vector with length $C$ and $C$ is the number of possible classification labels. Finally, the node $v_{i}$ is labeled with the class having the maximum probability.

\begin{figure}
	\centering
	\includegraphics[width=0.8\linewidth]{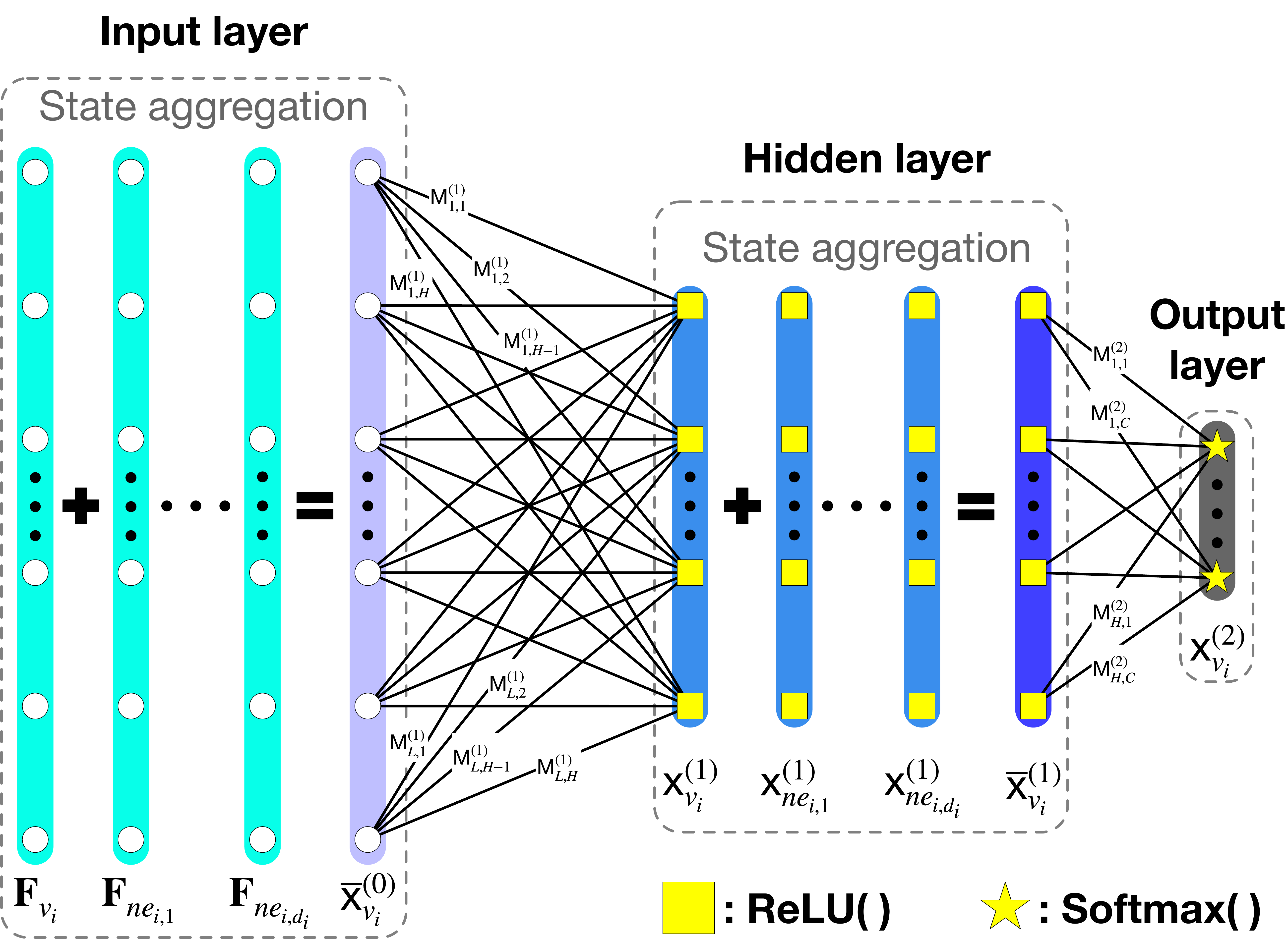}
	\caption{Illustration of the two-layer GCN model.}
	\label{fig:viewgcn}
	\vspace{-18pt}
\end{figure}

Without loss of generality and to facilitate the presentation, we elaborate on a representative two-layer GCN model \cite{GCN} as follows, and will use it to illustrate the design of our {\secgnn} afterwards. The GCN's propagation model is:
\begin{equation}
	\label{eq:GCN}
	\mathbf{Z}=\mathrm{Softmax}(\hat{\mathbf{A}}\mathrm{ReLU}(\hat{\mathbf{A}}\mathbf{F}\mathbf{M}^{(1)})\mathbf{M}^{(2)}),
\end{equation}
where $\mathbf{M}^{(1)}$ and $\mathbf{M}^{(2)}$ are two trainable weight matrices. $\hat{\mathbf{A}}$ is a symmetric normalized adjacency matrix $\hat{\mathbf{A}}=\tilde{\mathbf{D}}^{-\frac{1}{2}}\tilde{\mathbf{A}}\tilde{\mathbf{D}}^{-\frac{1}{2}}$
where $\tilde{\mathbf{A}}=\mathbf{A}+\mathbf{I}$ is the adjacency matrix of the graph with self-connection added ($\mathbf{I}$ is the identity matrix). $\tilde{\mathbf{D}}$ is a diagonal matrix:
\begin{equation}
	\label{eq:degree_matrix}
	\tilde{\mathbf{D}}_{i,i}=sw_{v_{i}}=\sum_{j\in[1,N]}\tilde{\mathbf{A}}_{i,j}=1+\sum_{j\in[1,d_{i}]}w_{i,j},
\end{equation}
where $N$ is the number of nodes in the graph, $d_{i}$ is the degree of node $v_{i}$ and $sw_{v_{i}}$ is the sum of $v_{i}$'s edge weights. Namely, $\tilde{\mathbf{D}}_{i,i}$ is the sum of node $v_{i}$'s edge weights with self-connection added (i.e, $w_{i,i}=1$). 
In particular, for a binary graph we have $\tilde{\mathbf{D}}_{i,i}=d_{i}+1$.
The activation function $\mathrm{ReLU}(x)$ is defined as \cite{nair2010rectified}:
\begin{equation}
	\mathrm{ReLU}(x)=\begin{cases}
		x & \text{if } x \ge 0,\\
		0 & \text{if } x < 0,\\
	\end{cases}
\end{equation}
and the activation function $\mathrm{Softmax}(\mathbf{x})$ is defined as \cite{liu2016large}:
\begin{equation}
	\label{eq:Softmax}
	z_{i}=\frac{e^{x_{i}}}{\sum_{j\in[1,C]}e^{x_{j}} }, i\in[1,C],
\end{equation}
where $C$ is the number of possible classification labels.

To \textit{train} the GCN model, the \textit{forward propagation} (i.e., Eq. \ref{eq:GCN}) is performed for each labeled node, and then the two trainable weight matrices $\mathbf{M}^{(1)}$ and $\mathbf{M}^{(2)}$ are updated based on the difference between each labeled node's inference result and label vector through backward propagation. Fig. \ref{fig:viewgcn} illustrates the process of performing the forward propagation for a node $v_{i}$:
\begin{enumerate}
	\item The $0_{th}$-layer \textit{aggregate state} $\overline{\mathbf{x}}_{v_{i}}^{(0)}$ of node $v_{i}$ is the weighted sum of the states of its neighbors and its own:
	\begin{equation}
		\label{eq:aggre1}
		\overline{\mathbf{x}}_{v_{i}}^{(0)}=(\hat{\mathbf{A}}\mathbf{F})_{v_{i}}=\hat{\mathbf{A}}_{v_{i},v_{i}}\mathbf{F}_{_{v_{i}}}+\sum_{j\in[1,d_{i}]}\hat{\mathbf{A}}_{v_{i},ne_{i,j}}\mathbf{F}_{ne_{i,j}}, 
	\end{equation} 
where $(\hat{\mathbf{A}}\mathbf{F})_{v_{i}}$ is the vector in row $v_{i}$ of matrix $\hat{\mathbf{A}}\mathbf{F}$, $\{\mathbf{F}_{ne_{i,j}}\}_{j\in[1,d_{i}]}$ are the feature vectors of $v_{i}$'s neighbors and $\hat{\mathbf{A}}_{v_{i},ne_{i,j}}$ is the element in row $v_{i}$ and column $ne_{i,j}$ of matrix $\hat{\mathbf{A}}$.

	\item 
	The $1_{st}$-layer state $\mathbf{x}_{v_{i}}^{(1)}$ of node $v_{i}$ is 
	\begin{equation}
	\label{eq:ac_ReLU}
		\mathbf{x}^{(1)}_{v_{i}}=\mathrm{ReLU}(\overline{\mathbf{x}}_{v_{i}}^{(0)}\mathbf{M}^{(1)}).
	\end{equation}
     \item The $1_{st}$-layer aggregate state of node $v_{i}$ is 
	\begin{equation}
		\label{eq:aggre2}
		\overline{\mathbf{x}}_{v_{i}}^{(1)}=(\hat{\mathbf{A}}\mathbf{X}^{(1)})_{v_{i}}=\hat{\mathbf{A}}_{v_{i},v_{i}}\mathbf{x}^{(1)}_{_{v_{i}}}+\sum_{j\in[1,d_{i}]}\hat{\mathbf{A}}_{v_{i},ne_{i,j}}\mathbf{x}^{(1)}_{_{ne_{i,j}}}, 
	\end{equation}
	where $\mathbf{X}^{(1)}$ is all nodes' $1_{st}$-layer states.

    \item The $2_{nd}$-layer state of node $v_{i}$ is 
    \begin{equation}
    \label{eq:ac_softm}
    \mathbf{Z}_{v_{i}}=\mathbf{x}^{(2)}_{v_{i}}=\mathrm{Softmax}(\overline{\mathbf{x}}_{v_{i}}^{(1)}\mathbf{M}^{(2)})
    \end{equation}
	which denotes the inference result of node $v_{i}$.
\end{enumerate}

After producing all labeled nodes' inference results through forward propagation, the average cross-entropy loss can be calculated by using all labeled nodes' labels and inference results:
\begin{equation}
	\label{eq:loss}
	\mathcal{L}=-\frac{1}{|\mathcal{T}|}\sum_{v_{i}\in \mathcal{T}}\sum_{j\in[1,C]}\mathbf{T}_{v_{i},j}\ln~ \mathbf{Z}_{v_{i},j},
\end{equation}
where $\mathcal{T}$ is the set of labeled nodes and $\mathbf{T}_{v_{i},j}$ is the classification label of node $v_{i}$, class $j$. Finally, each weight $\mathbf{M}_{i,j}\in\textbf{M}^{(1)}\cup \textbf{M}^{(2)} $ can be updated by its gradient:
\begin{equation}
 \mathbf{M}_{i,j}=\mathbf{M}_{i,j}-\rho\frac{\partial\mathcal{L}}{\partial \mathbf{M}_{i,j}},\notag
\end{equation}
 where $\rho$ is the learning rate. After the GCN model is trained, the classification label of each unlabeled node can be inferred through the forward propagation process.

\subsection{Additive Secret Sharing}
\label{sec:ss}
The 2-out-of-2 additive secret sharing of a secret value $x$ is denoted as $\llbracket x \rrbracket$, which can have the following two types \cite{mohassel2017secureml}: 
\begin{itemize}
	\item \textit{Arithmetic sharing}: $\llbracket x \rrbracket^{A}=\langle x\rangle_{1}+\langle x\rangle_{2}$ where $x,\langle x\rangle_{1}, \langle x\rangle_{2} \in  \mathbb{Z}_{2^{k}}$, and $\langle x\rangle_{1}, \langle x\rangle_{2}$ held by two parties, respectively.

	\item \textit{Binary sharing}: $\llbracket b \rrbracket^{B}=\langle b\rangle_{1}\oplus \langle b\rangle_{2}$ where $b,\langle b\rangle_{1}, \langle b\rangle_{2} \in  \mathbb{Z}_{2}$, and $\langle b\rangle_{1}, \langle b\rangle_{2}$ held by two parties, respectively.
\end{itemize}

The basic operations in the secret sharing domain under a two-party setting are as follows. (1) \emph{Linear operations.} Linear operations on secret-shared values only require local computation. In arithmetic sharing, if $\alpha,\beta,\gamma$ are public constants and $\llbracket x \rrbracket^{A}$, $\llbracket y \rrbracket^{A}$ are secret-shared values, then 
\begin{equation}
\llbracket \alpha x +\beta y+\gamma\rrbracket^{A}= (\alpha \langle x\rangle_{1}+\beta \langle y\rangle_{1}+\gamma, \alpha \langle x\rangle_{2}+\beta \langle y\rangle_{2}).\notag
\end{equation}
Each party can compute their respective shares locally based on the secrets they hold. 
(2) \emph{Multiplication.} Multiplication on secret-shared values requires one round of online communication. To multiply two secret-shared values: $\llbracket z \rrbracket^{A}=\llbracket x \rrbracket^{A} \times\llbracket y \rrbracket^{A}$, the two parties should first share a Beaver triple $\llbracket w \rrbracket^{A}=\llbracket u \rrbracket^{A} \times\llbracket v\rrbracket^{A}$ in the offline phase. After that, the party $P_{i\in\{0,1\}}$ locally computes $\langle e\rangle_{i}=\langle x\rangle_{i}-\langle u\rangle_{i}$ and $\langle f\rangle _{i}=\langle y\rangle _{i}-\langle v\rangle _{i}$, and then opens $e, f$ to each other. Finally, $P_{i}$ holds $
\langle z\rangle _{i}=i\times e\times f+f\times \langle u\rangle _{i}+e\times\langle v\rangle _{i}+ \langle w\rangle _{i}
$.

In binary sharing, the operations are similar to  arithmetic sharing. In particular, the addition operation is replaced by the XOR ($\oplus$) operation and multiplication is replaced by the AND ($\otimes$) operation. 
\vspace{-10pt}
\section{Problem Statement}
\label{sec:problem-statement}

\subsection{System Architecture}
\label{subsec:system-architecture}

There are two kinds of entities in {\secgnn}: the data owner and the cloud. The data owner (e.g., an online shopping enterprise or a social media service provider) wants to leverage the power of cloud computing to train a GNN model over his proprietary graph data as well as provide on-demand inference services once the model is trained.
Due to privacy concerns and that the graph data is proprietary, it is demanded that security must be embedded in the outsourced service, safeguarding the graph data, the trained model, as well as the inference results along the whole service flow.
The cloud providing the secure GNN training and inference is split into three cloud servers $P_{\{1,2,3\}}$ which can be operated by independent cloud service providers (e.g., AWS, Google, and Microsoft) in practice. 
Such multi-server model has also gained increasing traction in prior works on building efficient secure systems for other application domains \cite{ZhengDW18,chen2020metal,dauterman2020dory,araki2021secure,boneh2021lightweight,tan2021cryptgpu,dauterman2022waldo,wang2022privacy,bell2022distributed,wang2022oblivgm,zheng2022secskyline,wang2022PeGraph}.
In addition to the adoption in academia, such multi-server model has also been deployed  in industry. 
For example, Mozilla provides a service of lightweight private collection of telemetry data about Firefox under the non-colluding multi-server model \cite{Mozilla}; Apple and Google cooperatively provide  automated alerts about potential COVID-19 exposure to users, while providing strong privacy protections \cite{WhitePaper}.
{\secgnn} also follows such trend and contributes a new design for enabling privacy-preserving training and inference of GNNs in the cloud.

From a high-level point of view, the data owner in SecGNN will encrypt the graph by \emph{adequately} splitting the {\graph} into secret shares under 2-out-of-2 additive secret sharing, as per our design.
The secret shares are sent to $P_{1}$ and $P_{2}$, respectively. 
Upon receiving the encrypted {\graph}, $P_{\{1,2,3\}}$ perform our {\secgnn} to train the encrypted GNN model in the secret sharing domain.
Once the encrypted GNN model is trained, the data owner can query the cloud service to obtain encrypted classification labels for unlabeled nodes for decryption.
It is noted that the major computation in SecGNN is undertaken by the cloud servers $P_{1}$ and $P_{2}$ while $P_{3}$ provides necessary assistance, so as to simplify the interactions (and so the system implementation and deployment) as much as possible.

\vspace{-10pt}
\subsection{Threat Model}
\label{subsec:threat-model}

Similar to prior security designs in the three-server setting \cite{wagh2019securenn,MohasselRR20,tan2021cryptgpu}, we consider a semi-honest adversary setting where each of the three cloud servers honestly follow our protocol, but may \emph{individually} attempt to learn the private information of the data owner. 
The rationality of the non-collusion assumption is that the cloud service providers hosting the three cloud servers are normally business-driven and well-established parties, who are thus unwilling to risk their valuable commercial reputation by colluding with each other to intentionally breach data privacy \cite{WangWHZR16,chun2014outsourceable,qin2014towards}.
We consider that the data owner wishes to keep the following information private: (i) the features $\mathbf{F}$ and labels $\mathbf{T}$ of nodes, (ii) the adjacency matrix $\mathbf{A}$ encoding the structural information regarding the neighboring nodes of each node, the number of neighbors of each node, and the edge weight between each pair of connected nodes, (iii) the model weights $\mathbf{M}^{(1)}$ and $\mathbf{M}^{(2)}$, and (iv) the inference results for (unlabeled) nodes.

\vspace{-10pt}
\section{Secure GNN Training and Inference}
\label{sec:main}

\begin{figure}[t!]
	\centering
	\includegraphics[width=0.8\linewidth]{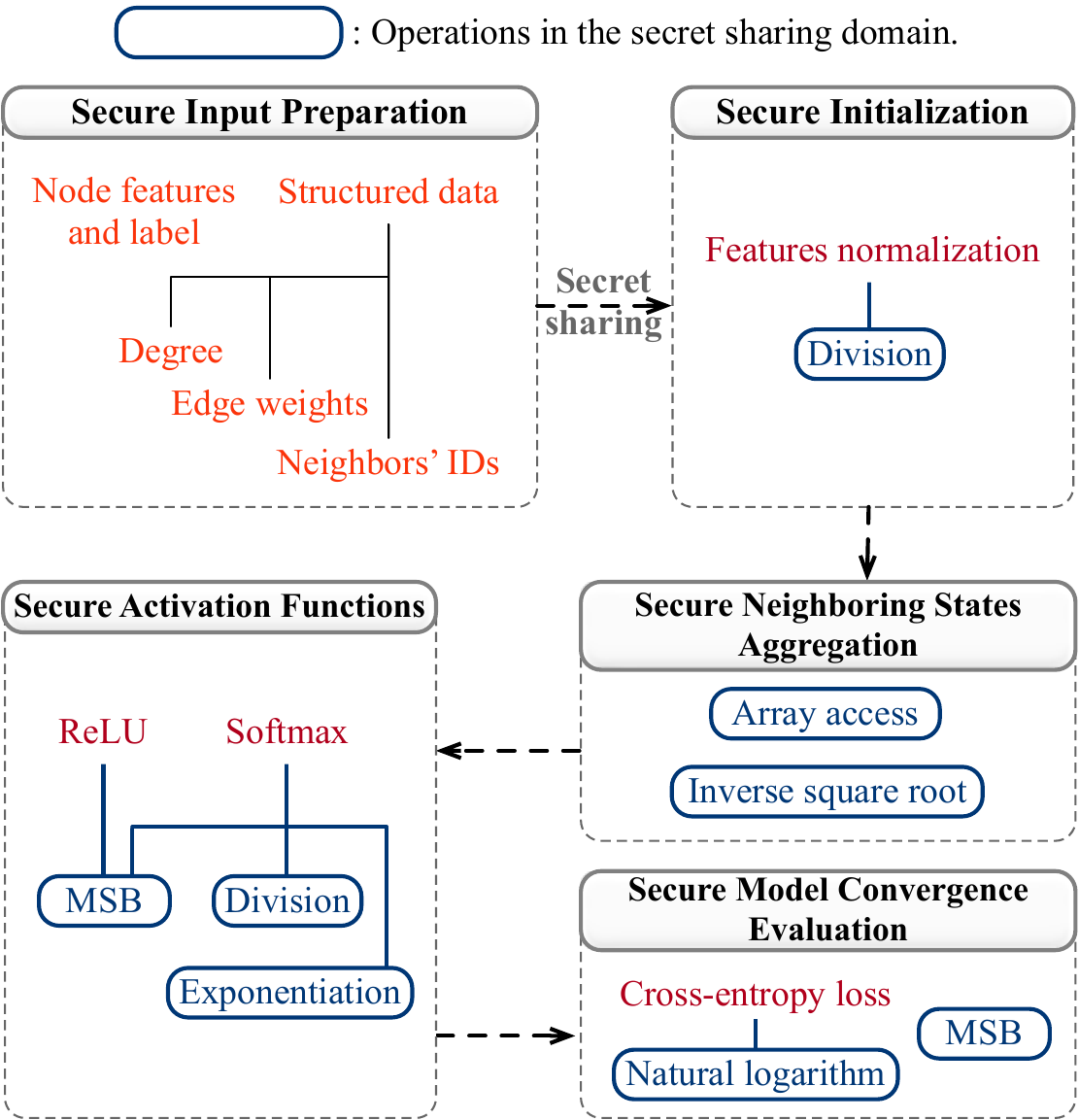}
	\caption{Overview of the core components in {\secgnn}.}
	\label{fig:overview}
	\vspace{-15pt}
\end{figure}

\subsection{{\secgnn} Overview}

Without loss of generality, we will use the two-layer GCN in Eq. \ref{eq:GCN} to illustrate the design of secure training and inference in {\secgnn}. 
Fig. \ref{fig:overview} provides an overview of the core components in SecGNN.
We will start with designing a \textit{secure input preparation} method, which allows the data owner to adequately encrypt its {\graph} so that they can support secure training and inference at the cloud. Subsequently, we design the following essential components to support the secure training and inference procedure at the cloud: (i) \textit{secure initialization} where the cloud normalizes the encrypted features for each node, (ii) \textit{secure neighboring states aggregation} where the cloud computes the encrypted aggregate state (as shown in Eq. \ref{eq:aggre1} and Eq. \ref{eq:aggre2}) for each node, (iii) \textit{secure activation functions} where the cloud activates the encrypted aggregate state for each node,  and (iv) \textit{secure model convergence evaluation} where the cloud performs a secure and fine-grained protocol to evaluate the convergence of the training process. Finally, we will elaborate on how to bridge the designed secure components to give the complete protocol for secure GCN training and inference.
\vspace{-10pt}
\subsection{Secure Input Preparation}
\label{subsec:sec-input-preparation}

\noindent\textbf{Encrypting node features and labels.} Given each node $v_{i}$'s initial feature vector with length $L$: $\mathbf{F}_{v_{i}}\in \mathbb{Z}_{2^{k}}^{L}$, the data owner generates a random vector $\mathbf{r}\in  \mathbb{Z}_{2^{k}}^{L}$. Then the arithmetic ciphertext of $\mathbf{F}_{v_{i}}$ is the secret shares $\langle \mathbf{F}_{v_{i}} \rangle_{1}=\{(\mathbf{F}_{v_{i},s}-\mathbf{r}_{s})\mod \mathbb{Z}_{2^{k}}\}_{s=1}^{L}$ and $\langle \mathbf{F}_{v_{i}}\rangle_{2}=\{\mathbf{r}_{s}\}_{s=1}^{L}$ where $\langle \mathbf{F}_{v_{i}}\rangle_{j}$ is sent to $P_{j}, j\in\{1,2\}$. Similarly, the data owner splits each labeled node's label vector $\mathbf{T}_{v_{i}}$ into secret shares.

\noindent\textbf{Encrypting structural information.} The structural information includes 1) each node's degree $d_{i}$; 2) the neighbors' IDs $ne_{i,j}$ of all nodes; 3) edge weights $w_{i,j}$ between all connected nodes. To protect the structural information, a simple method is to split the adjacency matrix $\mathbf{A}$ into secret shares. However, this method is inefficient and unnecessary since the adjacency matrix $\mathbf{A}$ is usually sparse. 

Instead, our insight is to devise a  set of data structures to properly store and represent the necessary structural information so that they can be encrypted efficiently as well as be used for GCN training and inference. In particular, we represent the structural information with an array-like data structure where each array element refers to a node's neighbor ID list and an edge weight list, and the array index is the node's ID. 

It is noted that as the degrees of nodes are different, the length of nodes' neighbor ID lists varies.
To protect each node's degree, the data owner pads several dummy neighbors' IDs to each node's neighbor ID list so that all nodes have the same number of neighbors. Namely, the secure neighbor ID list of $v_{i}$ is 
\begin{equation}
\mathbf{Ne}_{v_{i}}=\{ne_{i,1},\cdots, ne_{i,d_{i}}\}\cup \{ne_{i,1}',\cdots, ne_{i,d_{max}-d_{i}}'\},\notag
\end{equation}
where $ne_{}'$ are dummy neighbors' IDs, $d_{i}$ is $v_{i}$'s degree and $d_{max}$ is the maximum degree in the graph. However, if these dummy neighbors' IDs point to nodes that do not exist in the graph, the cloud servers will distinguish them from $\mathbf{Ne}_{v_{i}}$ when accessing these dummy neighbors, while if the dummy neighbors' IDs point to real nodes in the graph, the accuracy of the trained model will be degraded dramatically since dummy neighbors' states will change node $v_{i}$'s aggregate state.

Our solution is based on the observation that in GCN or GNN, a node's aggregate state is the weighted sum of its neighboring states, where the weights are relevant with $v_{i}$'s edge weights $w_{i,j}$. Therefore, we can set the edge weights between $v_{i}$ and its dummy neighbors to 0. Namely, $v_{i}$'s secure edge weight list is $\mathbf{W}_{v_{i}}=\{w_{i,1},\cdots, w_{i,d_{i}}\}\cup \{0_{j}\}_{j=1}^{d_{max}-d_{i}}$.
By this way, the effect of the dummy neighbors will be eliminated, which will be understood clearly in Section \ref{sec:aggre}. 
After padding dummy neighbors, the data owner splits each node's secure neighbor ID list $\mathbf{Ne}_{v_{i}}$ and secure edge weight list $\mathbf{W}_{v_{i}}$ into secret shares:
\begin{align}\notag
		\llbracket \mathbf{Ne}_{v_{i},j}\rrbracket^{A}&=\langle \mathbf{Ne}_{v_{i},j}\rangle_{1}+\langle \mathbf{Ne}_{v_{i},j}\rangle_{2}, j\in [1,d_{max}],\\
		\llbracket \mathbf{W}_{v_{i},j}\rrbracket^{A}&=\langle \mathbf{W}_{v_{i},j}\rangle_{1}+\langle \mathbf{W}_{v_{i},j}\rangle_{2}, j\in [1,d_{max}],\notag
\end{align}
where $i\in[1,N],j\in[1,d_{max}]$. Finally, the data owner sends all secret shares $\llbracket \mathbf{F}\rrbracket^{A}, \llbracket \mathbf{T}\rrbracket^{A},\llbracket \mathbf{Ne}\rrbracket^{A}$ and $\llbracket \mathbf{W}\rrbracket^{A} $ to $P_{1}$ and $P_{2}$, respectively. 
Assuming that the nodes in the graph are indexed from $1$ to $N$, i.e., $v_{1}=1, \cdots, v_{N}=N$.
The encrypted {\graph} can be regarded as an array-like data structure where each array element is a node's encrypted data and the index is the node's ID $v_{i}$.

\vspace{-10pt}
\subsection{Secure Initialization}
\label{sec:initia}

Each node's initial features need to be normalized before model training \cite{glorot2010understanding}. Without loss of generality, we will work with a common feature normalization method:
\begin{equation}
	\label{eq:norm}
	\overline{x}_{i}=\frac{x_{i}}{\sum_{j\in[1,L]}x_{j}},i\in[1,L],
\end{equation}
where $L$ is the number of features. Obviously, the \textit{sum} operation is directly supported in the secret sharing domain, but the \textit{division} operation is hard to be directly supported and calls for a tailored protocol for secure division in the secret sharing domain.

Our solution is to approximate the division operation using basic operations (i.e., $+, \times$) supported in the secret sharing domain. We observe that the main challenge in computing division is to compute the reciprocal $\llbracket\frac{1}{  x}\rrbracket^{A}$. Inspired by the recent work \cite{knott2020crypten}, we approximate the reciprocal by the iterative Newton-Raphson algorithm \cite{akram2015newton}:
\begin{equation}
	\label{recip}
 y_{n+1} = y_{n} (2- x y_{n} ), 
\end{equation}
which will converge to $  y_{n}\approx\frac{1}{x }$. Obviously, both subtraction and multiplication are naturally supported in the secret sharing domain. In addition, a faster convergence can be achieved by initializing $ y_{0}$ as:
\begin{equation}
	\label{initial_recip}
 y_{0}=3e^{0.5- x}+0.003.
\end{equation}
How to compute $e^{ x}$ in the secret sharing domain will be introduced in Section \ref{sec:Softmax}. Subroutine \ref{alg:a1} describes our protocol for secure feature normalization.
\begin{algorithm}[!t]
\caption{Secure Feature Normalization} 
\label{alg:a1}
\begin{algorithmic}[1] 
	\REQUIRE Node $v_{i}$'s encrypted features $\{ \llbracket \mathbf{F}_{v_{i},j}\rrbracket^{A}\}_{j\in[1,L]}$.
	
	\ENSURE Encrypted normalized features$\{\llbracket  \overline{\mathbf{F}}_{v_{i},j}\rrbracket^{A}\}_{j\in[1,L]}$.
	\STATE $P_{\{1,2\}}$ locally calculate $\llbracket  S\rrbracket^{A}=\sum_{l\in[1,L]}	\llbracket \mathbf{F}_{v_{i},l}\rrbracket^{A}$.
	
	//$P_{\{1,2\}}$ calculate the approximate $\llbracket \frac{1}{ S}\rrbracket^{A}$:
	\STATE $\llbracket y_{0} \rrbracket^{A}=3\times \llbracket e^{0.5- S}\rrbracket^{A}+0.003$;
	\FOR{$n=0$ to $\mathcal{N}$} 
	\STATE  $\llbracket y_{n+1} \rrbracket^{A}=\llbracket y_{n} \rrbracket^{A}\times (2-\llbracket S\rrbracket^{A}\times \llbracket y_{n} \rrbracket^{A})$.
	\ENDFOR 
	
	//$P_{\{1,2\}}$ calculate the normalized features:
	\FOR{$l=1$ to $L$} 
	\STATE $\llbracket \overline{\mathbf{F}}_{v_{i},l}\rrbracket^{A}=\llbracket \mathbf{F}_{v_{i},l}\rrbracket^{A} \times \llbracket \frac{1}{ S}\rrbracket^{A}$.
	\ENDFOR 
\end{algorithmic}
\end{algorithm}
 \vspace{-10pt}
\subsection{Secure Neighboring States Aggregation}
\label{sec:aggre}

During the state propagation process, the $k_{th}$-layer aggregate state $\overline{\mathbf{x}}_{v_{i}}^{(k)}$ of node $v_{i}$ is computed by node $v_{i}$'s $k_{th}$-layer state $\mathbf{x}_{v_{i}}^{(k)}$ and its neighbors' $k_{th}$-layer states $\mathbf{x}_{\mathbf{Ne}_{v_{i},j}}^{(k)}$, where node $v_{i}$'s $0_{th}$-layer state is its normalized feature vector $\overline{\mathbf{F}}_{v_{i}}$.
As shown in Eq. \ref{eq:aggre1} and Eq. \ref{eq:aggre2}, the aggregate state is the weighted sum of these states. However, since only the encrypted neighbors' IDs are uploaded to the cloud rather than the whole adjacency matrix $\mathbf{A}$, it raises a challenge on how to compute $\hat{\mathbf{A}}$ and perform all subsequent operations. 

Our insight is to first transform the aggregation method in Eq. \ref{eq:aggre1} and Eq. \ref{eq:aggre2} to the other form that can be calculated in the secret sharing domain. The $k_{th}$-layer aggregate state of node $v_{i}$ can be denoted as $(\tilde{\mathbf{D}}^{-\frac{1}{2}}\tilde{\mathbf{A}}\tilde{\mathbf{D}}^{-\frac{1}{2}}\mathbf{X}^{(k)})_{i}$ where $()_{i}$ denotes the $i_{th}$ row of the matrix, and its equivalent form is $\sum_{j=1}^{N}\frac{\tilde{\mathbf{A}}_{i,j}}{\sqrt{\tilde{\mathbf{D}}_{i,i}\tilde{\mathbf{D}}_{j,j}}}\mathbf{X}^{(k)}_{j}$, where $\sum_{k=1}^{N}\tilde{\mathbf{D}}_{i,k}^{-\frac{1}{2}}=\tilde{\mathbf{D}}_{i,i}^{-\frac{1}{2}}$ because $\tilde{\mathbf{D}}$ is a diagonal matrix. Since $\tilde{\mathbf{A}}_{i,j}=0$ if node $v_{j}$ is not $v_{i}$'s neighbors, $w_{i,i}=1$ and $sw_{v_{i}}=\tilde{\mathbf{D}}_{v_{i},v_{i}}$ (i.e., Eq. \ref{eq:degree_matrix}), a more simper form of $v_{i}$'s $k_{th}$-layer aggregate state is:
\begin{equation}
	\label{aggre}
	\overline{\mathbf{x}}^{(k)}_{v_{i}}=\frac{1}{sw_{v_{i}}}\mathbf{x}^{(k)}_{v_{i}}+\sum_{j\in[1,d_{max}]}\frac{\mathbf{W}_{v_{i},j}}{\sqrt{sw_{v_{i}}}\cdot \sqrt{ sw_{\mathbf{Ne}_{v_{i},j}}}}\mathbf{x}^{(k)}_{\mathbf{Ne}_{v_{i},j}}.
\end{equation}
It is noted that since the edge weights between $v_{i}$ and its dummy neighbors are $0$, the effect of these dummy neighbors can be eliminated using Eq. \ref{aggre}. 

When securely computing node $v_{i}$'s $k_{th}$-layer aggregate state $\overline{\mathbf{x}}_{v_{i}}^{(k)}$ by Eq. \ref{aggre}, the cloud servers should first securely access the neighboring nodes' $k_{th}$-layer states $\mathbf{x}^{(k)}_{\mathbf{Ne}_{v_{i},j}},j\in[1,d_{max}]$, and then sum these states by securely multiplying its weight $\frac{\mathbf{W}_{v_{i},j}}{\sqrt{sw_{v_{i}}}\cdot \sqrt{ sw_{\mathbf{Ne}_{v_{i},j}}}}$. However, it is challenging to access the neighboring nodes' states since the neighbors' IDs are encrypted. Meanwhile, the square root is not naturally supported in the secret sharing domain.

To overcome the two obstacles, we design a protocol for \textit{secure neighboring states access} which allows the cloud servers to securely access the neighboring nodes' states, and a protocol for \textit{secure neighboring states summation} allowing the cloud servers to securely perform the square root calculation and the summation of the accessed neighboring states. 
\vspace{-10pt}
\subsubsection{Secure Neighboring States Access}
\label{sec:access}
Neighboring states access is challenging in the secret sharing domain, because we need to access each neighbor's state with both the neighbor's ID $\mathbf{Ne}_{v_{i},j}$ and state $\mathbf{x}^{(k)}_{\mathbf{Ne}_{v_{i},j}}$ being encrypted.
Furthermore, the accessed result should still be encrypted. 
Our insight is to first transform it to the array access problem in the secret sharing domain, i.e., the state vector $\mathbf{x}^{(k)}_{v_{i}}$ of each node in the graph is treated as an array element and node IDs ($1$ to $N$) serve as array indexes. We then consider how to securely access the \textit{encrypted element} at the \textit{encrypted location} from the \textit{encrypted array}.

From the literature, we identify the existence of the state-of-the-art secure array access protocol in the secret sharing domain by Blanton \textit{et al.} \cite{blanton2020improved}, which works in a similar three-party setting and uses 2-out-of-2 secret sharing. This method requires communicating $\mathbf{4m+4}$ elements in \textbf{two} rounds, where $m$ is the length of the encrypted array. In the protocol of \cite{blanton2020improved}, the cloud needs to send random values to each other during accessing the encrypted array element. These shared random values will be used to hide the shares of each array element, and will be offset in the sum of shares. 

Through careful inspection on the protocol, we manage to design a more efficient protocol which only requires communicating $\textbf{2m+2}$ elements in \textbf{one} round. In particular, instead of letting the cloud servers send random values to each other, our idea is to enable them to locally generate \textit{correlated random values} (i.e., $c^{1}+c^{2}+c^{3}=0$) based on a technique from \cite{araki2016high}, which will be used to hide the shares of each array element, and will be offset in the sum of shares. More specifically, in the system initialization phase, the cloud server $P_{i},i\in\{1,2,3\}$ samples a key $k_{i}$ and send $k_{i}$ to $P_{i+1}$ where $P_{3+1=1}$. Then $P_{i}$'s $j_{th}$ correlated random value is
 \begin{equation}
\notag
\mathbf{c}^{i}[j]=\mathbb{F}(k_{i},j)-\mathbb{F}(k_{i-1},j),
\end{equation}
where $k_{1-1=3}$ and $\mathbb{F}$ is a pseudorandom function (PRF). Meanwhile, an agreed random value $r$ between each two cloud servers can also be generated by their shared key. The $j_{th}$ agreed random value between $P_{i}$ and $P_{i+1}$ is $r_{j}^{i}=\mathbb{F}(k_{i},j) \mod m$,
where $m$ is the length of the secret array.

Given a secret array $\llbracket \mathbf{a}\rrbracket^{A} =\langle \mathbf{a}\rangle_{1}+\langle \mathbf{a}\rangle_{2}$ and a secret index $\llbracket I\rrbracket^{A}= \langle  I\rangle_{1} +\langle  I\rangle_{2}$ held by $P_{1}$ and $P_{2}$, respectively, our protocol, as shown in Subroutine \ref{alg:a2}, for securely accessing the element $ \llbracket \mathbf{a}[I]\rrbracket^{A}$ is as follows:
\begin{enumerate}
	\item   $P_{1}$ first rotates its shares $r^{1}$ locations:
	\begin{equation}
		\langle \mathbf{a}[1] \rangle_{1},\cdots, \langle \mathbf{a}[m] \rangle_{1} ~~~	\circlearrowright\notag \blacktriangledown\notag
	\end{equation}
	\begin{equation}
		\langle \mathbf{a}[m-r^{1}] \rangle{_{1}},\cdots, 	\langle \mathbf{a}[m] \rangle_{1}, \langle \mathbf{a}[1] \rangle_{1},\cdots ,\langle \mathbf{a}[m-r^{1}+1] \rangle_{1}. \notag
	\end{equation}
	Then, $P_{1}$ sets the new array as $	\langle   \mathbf{a}'[j] 	\rangle_{1}=	\langle   \mathbf{a}[j] \rangle_{1}+\mathbf{c}^{1}[j], j\in[1,m]$, and rotates it $r^{3}$ locations:
		\begin{equation}
		\langle \mathbf{a}'[1] \rangle_{1},\cdots, \langle \mathbf{a}'[m] \rangle_{1} ~~~	\circlearrowright\notag \blacktriangledown\notag
	\end{equation}
	\begin{equation}
		\langle \mathbf{a}'[m-r^{3}] \rangle_{1},\cdots, 	\langle \mathbf{a}'[m] \rangle_{1}, \langle \mathbf{a}'[1] \rangle_{1},\cdots, \langle \mathbf{a}'[m-r^{3}+1] \rangle_{1}. \notag
	\end{equation}
	The new array is denoted as  $\langle  \mathbf{a}'' \rangle _{2}$.  Finally, $P_{1}$ sets $ \langle h\rangle _{1}=(\langle I \rangle_{1}+r^{1}+r^{3})\mod m$, then sends $ \langle h\rangle _{1}$ and $\langle  \mathbf{a}''\rangle _{2}$ to $P_{2}$. 

	\item 
	$P_{2}$ sets $h=(\langle h\rangle _{1}+\langle I\rangle_{2} )\mod m$, then $P_{2}$'s share of the accessed element $\mathbf{a}[I]$ is $\underline{\langle  \mathbf{a}''[h] \rangle _{2}}$.
	
	\item $P_{2}$ first rotates its shares of the raw array $r^{1}$ locations:
	\begin{equation}
	\langle \mathbf{a}[1] \rangle_{2},\cdots, \langle \mathbf{a}[m] \rangle_{2} ~~~	\circlearrowright\notag \blacktriangledown\notag
	\end{equation}
	\begin{equation}
	\langle \mathbf{a}[m-r^{1}] \rangle_{2},\cdots, 	\langle \mathbf{a}[m] \rangle_{2}, \langle \mathbf{a}[1] \rangle_{2},\cdots, \langle \mathbf{a}[m-r^{1}+1] \rangle_{2}. \notag
	\end{equation}
    Then, $P_{2}$ sets the new array as $ \langle   \mathbf{a}'[j]	\rangle_{2}=	\langle   \mathbf{a}[j] \rangle_{2}+\mathbf{c}^{2}[j], j\in[1,m]$, and sends $ \langle   \mathbf{a}'\rangle_{2}$ and $h$ to $P_{3}$. 
    
	\item $P_{3}$ first sets $\langle \mathbf{ a}''[j] 	\rangle_{3}=	\langle   \mathbf{a}'[j]\rangle_{2}+\mathbf{c}^{3}[j], j\in[1,m]$, then rotates them $ r^{3}$ locations: 
	\begin{equation}
	\langle \mathbf{a}''[1] \rangle_{3},\cdots, \langle \mathbf{a}''[m] \rangle_{3}~~~	\circlearrowright \blacktriangledown\notag
	\end{equation}
	\begin{align}
	&\langle \mathbf{a}''[m-r^{3}] \rangle_{3},\cdots, 	\langle \mathbf{a}''[m] \rangle_{3},\\\notag
	&\langle \mathbf{a}''[1] \rangle_{3},\cdots ,\langle \mathbf{a}''[m-r^{3}+1] \rangle_{3}. \notag
	\end{align}
   Finally, $P_{3}$'s share of $\mathbf{a}[I]$ is $\underline{\langle  \mathbf{a}''[h] \rangle _{3}}$.
\end{enumerate}

\begin{algorithm}[!t]
	\caption{Secure Array Access} 
	\label{alg:a2}
	\begin{algorithmic}[1] 
		\REQUIRE  $P_{\{1,2\}}$ hold the secret shares $\langle \mathbf{a}\rangle_{\{1,2\}}$ and $\langle I \rangle_{\{1,2\}}$, respectively;  
		$P_{\{1,2,3\}}$ hold random value array $\mathbf{c}^{\{1,2,3\}}$, respectively; $P_{\{1,2\}}$ and $P_{\{1,3\}}$ hold the agreed random values $r^{1}$ and $r^{3}$, respectively.
		
		\ENSURE $P_{\{2,3\}}$ hold the accessed element $\llbracket\mathbf{a}[I]\rrbracket^{A}$.
		
		// \underline{$P_{1}$ locally performs:}
		
		\STATE  $\langle\mathbf{a}\rangle_{1}=\langle\mathbf{a} \rangle_{1} \circlearrowright r^{1}$.  $\langle\mathbf{a}'\rangle_{1}= \langle\mathbf{a}\rangle_{1}+\mathbf{c}^{1}$. $\langle\mathbf{a}''\rangle_{2}=\langle\mathbf{a} '\rangle_{1} \circlearrowright r^{3}$.

		\STATE $ \langle h\rangle _{1}=(\langle I \rangle_{1}+r^{1}+r^{3})\mod m$.
		
		\STATE Sending $ \langle h\rangle _{1}$ and $\langle  \mathbf{a}''\rangle _{2}$ to $P_{2}$.
		
		// \underline{$P_{2}$ locally performs:}
			
		\STATE  $h=(\langle h\rangle _{1}+\langle I\rangle_{2} )\mod m$.
		
		\STATE $P_{2}$ sets the secret share of $\mathbf{a}[I]$ as $\underline{\langle  \mathbf{a}''[h] \rangle _{2}}$.
		
		\STATE $\langle\mathbf{a}\rangle_{2}=\langle\mathbf{a} \rangle_{2} \circlearrowright r^{1}$. $\langle\mathbf{a}'\rangle_{2}= \langle\mathbf{a}\rangle_{2}+\mathbf{c}^{2}$.
		
		\STATE Sending $ h$ and $\langle\mathbf{a}'\rangle_{2}$ to $P_{3}$.
		
		// \underline{$P_{3}$ locally performs:}
			
		\STATE $ \langle\mathbf{a}''\rangle_{3}= \langle\mathbf{a}'\rangle_{2}+\mathbf{c}^{3}$. $\langle\mathbf{a}''\rangle_{3}=\langle\mathbf{a}'' \rangle_{3} \circlearrowright r^{3}$.

		\STATE  $P_{3}$ sets the secret share of $\mathbf{a}[I]$ as $ \underline{\langle  \mathbf{a}''[h] \rangle _{3}}$.
	\end{algorithmic}
\end{algorithm}

It is noted that, the correlated random values $\mathbf{c}^{\{1,2,3\}}$ and agreed random values $r^{\{1,3\}}$ all do not require online communication because they are generated by the PRF and shared keys. Therefore, our protocol only requires communicating $2m+2$ elements in one round, i.e., in steps 1), 3). 

\noindent\textbf{Correctness analysis.} $P_{2}$'s shares $\langle  \mathbf{a}'' \rangle_{2}$ are generated by $
	\langle  \mathbf{a}'' \rangle_{2}=(\langle  \mathbf{a} \rangle _{1} \circlearrowright r^{1}+\mathbf{c}^{1})\circlearrowright r^{3}$,
where "$\circlearrowright$" denotes "rotate". $P_{3}$'s shares $\langle  \mathbf{a}''\rangle _{3}$ are generated by $
	\langle  \mathbf{a}''\rangle _{3}=(\langle  \mathbf{a} \rangle _{2} \circlearrowright r^{1}+\mathbf{c}^{2}+\mathbf{c}^{3})\circlearrowright r^{3}$. Based on $\mathbf{c}^{1}[j]+\mathbf{c}^{2}[j]+\mathbf{c}^{3}[j]=0$, we can obtain $\langle  \mathbf{a}''\rangle _{2}+\langle  \mathbf{a}''\rangle _{3}=\mathbf{a}\circlearrowright r^{1}\circlearrowright r^{3}$, namely, for $j\in[1,m]$, $
\langle  \mathbf{a}''[(j+r^{1}+r^{3})\mod m]\rangle _{2}+
\langle  \mathbf{a}''[(j+r^{1}+r^{3})\mod m]\rangle _{3}=\mathbf{a}[j]$. Since $h=I+r^{1}+r^{3}$, the accessed element is exactly $\langle \mathbf{a}''[h] \rangle _{2}+ \langle \mathbf{a}''[h] \rangle _{3}=\mathbf{a}[I]$.

It is noted that since $P_{\{1,2\}}$ perform main computations in our protocol but $P_{3}$ holds the secret share of the accessed element, $P_{3}$ should re-share its secret $\langle \mathbf{a}''[h]\rangle_{3}$ to $P_{1}$ and $P_{2}$. More specifically, $P_{3}$ generates a random value $s$, and then sends $s, \langle \mathbf{a}''[h]\rangle_{3}-s$ to $P_{1}$ and $P_{2}$, respectively. Finally, the shares held by $P_{1}$ and $P_{2}$ are $s$ and $\langle \mathbf{a}''[h]\rangle_{2}+\langle \mathbf{a}''[h]\rangle_{3}-s$, respectively.

\vspace{-10pt}
\subsubsection{Secure Neighboring States Summation}
After performing the above secure neighboring states access protocol, $P_{\{1,2\}}$ hold all encrypted neighboring states $\mathbf{x}^{(k)}_{\mathbf{Ne}_{v_{i},j}}$ of node $v_{i}$. In addition, as shown in Eq. \ref{aggre}, the sum of $v_{i}$'s each neighbor's own edge weights $sw_{\mathbf{Ne}_{v_{i},j}}$ (i.e., Eq. \ref{eq:degree_matrix}) are used in calculating the aggregate state $\overline{\mathbf{x}}^{(k)}_{v_{i}}$. Since each neighbor's own edge weights are attached with its ID like its state, similar to accessing neighboring states, the cloud servers should access each neighbor's edge weights using the above secure array access protocol. After that, the cloud servers can obtain node $v_{i}$'s each neighbor's encrypted state $\mathbf{x}^{(k)}_{\mathbf{Ne}_{v_{i},j}}$ and the encrypted sum of edge weights $sw_{\mathbf{Ne}_{v_{i},j}}$ using each encrypted neighbor's ID. Then the cloud uses Eq. \ref{aggre} to calculate node $v_{i}$'s aggregate state $\overline{\mathbf{x}}_{v_{i}}^{(k)}$. However, the \textit{square root} is not naturally supported in secret sharing.

\begin{algorithm}[!t]
	\caption{Secure Neighboring States Summation} 
	\label{alg:a3}
	\begin{algorithmic}[1] 
		\REQUIRE  Node $v_{i}$'s $k_{th}$-layer state and edge weights:  $ \llbracket\mathbf{x}^{(k)}_{v_{i}} \rrbracket^{A}$ and $\llbracket \mathbf{W}_{v_{i},j} \rrbracket^{A},j\in[1,d_{max}]$, and $v_{i}$'s neighboring $k_{th}$-layer states and their edge weight sum: $\llbracket\mathbf{x}^{(k)}_{\mathbf{Ne}_{v_{i},j}} \rrbracket^{A}$ and $\llbracket sw_{\mathbf{Ne}_{v_{i},j}} \rrbracket^{A}, j\in[1,d_{max}]$. 
		
		\ENSURE $v_{i}$'s $k_{th}$-layer encrypted aggregate state $\llbracket\overline{\mathbf{x}}^{(k)}_{v_{i}}\rrbracket^{A}$.
		
		\STATE $P_{\{1,2\}}$ first calculate the approximate $\llbracket\frac{1}{sw_{v_{i}}} \rrbracket^{A}$ by Eq. \ref{recip}.
		
		//$P_{\{1,2\}}$ calculate each approximate $\llbracket\frac{1}{\sqrt{sw_{id}}} \rrbracket^{A}$:
		
		\FOR{each $id\in\{v_{i}\}\cup \{\mathbf{Ne}_{v_{i},j}\}_{j\in[1,d_{max}]}$} 
		
		\STATE $\llbracket y_{0} \rrbracket^{A}=3\times \llbracket e^{0.5- sw_{id}}\rrbracket^{A}+0.003$.
		
		\FOR{$n=0$ to $\mathcal{N}$} 
		\STATE  $\llbracket y_{n+1} \rrbracket^{A} =\frac{1}{2}\times \llbracket y_{n}\rrbracket^{A}\times(3-\llbracket sw_{id}\rrbracket^{A}\times\llbracket y_{n}\rrbracket^{A}\times\llbracket y_{n}\rrbracket^{A})$.
		\ENDFOR 
		\ENDFOR		
		
		//$P_{\{1,2\}}$ calculate the aggregate state:  
		
		\STATE $\llbracket \overline{\mathbf{x}}^{(k)}_{v_{i}}\rrbracket^{A}=\llbracket  \frac{1}{sw_{v_{i}}} \rrbracket^{A}\times \llbracket \mathbf{x}^{(k)}_{v_{i}} \rrbracket^{A}$.
		
		\FOR{$j=1$ to $d_{max}$} 
		\STATE $\llbracket \overline{\mathbf{x}}^{(k)}_{v_{i}} \rrbracket^{A}=\llbracket \overline{\mathbf{x}}^{(k)}_{v_{i}} \rrbracket^{A}+\llbracket \mathbf{W}_{v_{i},j} \rrbracket^{A}\times\llbracket  \frac{1}{\sqrt{sw_{v_{i}}}} \rrbracket^{A}\times\llbracket  \frac{1}{\sqrt{sw_{\mathbf{Ne}_{v_{i},j}}}} \rrbracket^{A}\times \llbracket \mathbf{x}^{(k)}_{\mathbf{Ne}_{v_{i},j}} \rrbracket^{A}$.
		\ENDFOR
	\end{algorithmic}
\end{algorithm}

Inspired by the very recent work \cite{knott2020crypten},  we resort to the approach of approximating the inverse square root by iterative Newton-Raphson algorithm \cite{akram2015newton}:
\begin{equation}
	\label{eq:square_root}
y_{n+1}=\frac{1}{2}y_{n}(3-xy_{n}^{2}), 
\end{equation}
which will converge to $y_{n}\approx \frac{1}{\sqrt{x}}$. Obviously, both subtraction and multiplication are naturally supported in the secret sharing domain. The initialization $y_{0}$ can be set as $y_{0}=3e^{0.5- x}+0.003$.

After securely accessing each neighboring node's state $\llbracket\mathbf{x}^{(k)}_{\mathbf{Ne}_{v_{i},j}} \rrbracket^{A}$ for node $v_{i}$, the cloud servers utilize the above secure inverse square root protocol to perform the secure neighboring states summation, as shown in Subroutine \ref{alg:a3}.

\subsection{Secure Activation Functions}
\label{sec:activation}

After a node $v_{i}$'s $k_{th}$-layer aggregate state $\overline{\mathbf{x}}^{(k)}_{v_{i}}$ is calculated and multiplied with the trainable weight matrix $\mathbf{M}^{(k+1)}$, i.e., $\overline{\mathbf{x}}^{(k)}_{v_{i}}\mathbf{M}^{(k+1)}$, an activation functions needs to be applied over $\hat{\mathbf{x}}^{(k)}_{v_{i}}=\overline{\mathbf{x}}^{(k)}_{v_{i}}\mathbf{M}^{(k+1)}$ to calculate $v_{i}$'s $(k+1)_{th}$-layer state $\mathbf{x}^{(k+1)}_{v_{i}}$, according to Eq. \ref{eq:ac_ReLU} and Eq. \ref{eq:ac_softm} respectively.
In this section, we will introduce how to securely compute the activation functions in the secret sharing domain.
\vspace{-10pt}
\subsubsection{Secure ReLU Function}
\label{sec:ReLU}
The function $\mathrm{ReLU}:=max(x,0)$ is a popular activation function in neural network, whose core is to test whether $x>0$ or not. However, the comparison operation is not naturally supported in the secret sharing domain. 
We note that given the computation is in $\mathbb{Z}_{2^k}$, it suffices to tailor a protocol for testing whether the Most Significant Bit (MSB) of $\llbracket x\rrbracket^{A} $ is $0$ or not \cite{mohassel2018aby3,LiuZYY21}. Mohassel et al. \cite{mohassel2018aby3} propose to compute the MSB using secure bit decomposition (only directions briefly mentioned without a concrete construction though). It is noted that different from our system, their security design uses replicated secret sharing, which runs among three cloud servers and needs them to interact with each other throughout the process. Inspired by their work, we provide an alternative design to evaluate the MSB under additive secret sharing that suits our system, in which the computation is mainly conducted by $P_{1}$ and $P_{2}$ while $P_{3}$ just provides necessary triples in advance. The details of our design are as follows.

\begin{figure}[t!]
	\centering
	\includegraphics[width=0.8\linewidth]{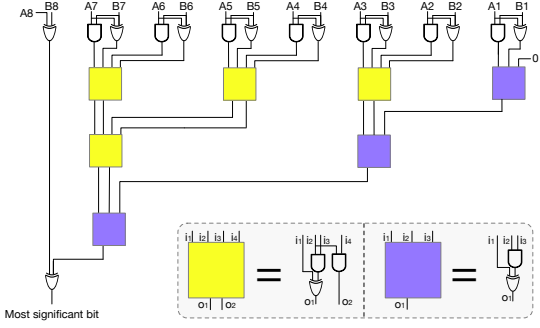}
	\caption{An 8-bit tailored PPA.}
	\label{fig:ppa}
	\vspace{-10pt}
\end{figure}

Given two fixed point numbers' complement $A$ and $B$, which can represent the shares of a secret value, the MSB of $A+B$ can be computed by a tailored Parallel Prefix Adder (PPA) \cite{harris2003taxonomy}. Fig. \ref{fig:ppa} illustrates an 8-bit tailored PPA. We can apply the tailored PPA to the secret shares.
In particular, given the $k$-bit secret sharing $\llbracket x \rrbracket^{A}=\langle x\rangle_{1}+\langle x\rangle_{2}$ held by $P_{1}$ and $P_{2}$, they first locally decompose the complement of $\langle x\rangle_{i}$ into bits: $\langle x\rangle_{i}=x_{i}[1],\cdots, x_{i}[k],i\in\{1,2\}$. After that, they input the bits into a $k$-bit tailored PPA to perform secure AND and XOR calculations. Given a $k$-bit number, the tailored PPA can calculate its MSB in $\mathrm{log}~ k$ rounds. In addition, as shown in Section \ref{sec:ss}, in additive secret sharing, a AND gate requires online communication 4 bits in one round, while an XOR gate does not require communication. Therefore, to calculate the MSB of a $k$-bit number in additive secret sharing, our alternative design requires the two cloud servers to online communicate $12k-12-4 \mathrm{log}~k$ bits in $\mathrm{log}~k$ rounds. 
It is noted that, using the above method, $msb(x)=1$ if $x<0$, and $msb(x)=0$ if $x>=0$. To be compatible with the subsequent operation, one of $P_{1}$ and $P_{2}$ flips its share $\langle msb(x)\rangle_{1}$ or $\langle msb(x)\rangle_{2}$ so that $msb(x)'=0$ if $x<0$, and $msb(x)'=1$ if $x>=0$.

However, using the above method, the cloud servers only obtain $\llbracket msb(x)'\rrbracket^{B}$, not $\mathrm{ReLU}(x)$, and the cloud servers also need to calculate $\llbracket msb(x)'\rrbracket^{B}\times\llbracket x\rrbracket^{A}$ when $P_{1}$ and $P_{2}$ hold $ \langle  msb(x)'\rangle_{1},\langle x\rangle_{1}$ and $ \langle  msb(x)' \rangle_{2},\langle x\rangle_{2}$, respectively. Inspired by \cite{mohassel2018aby3}, we design a tailored protocol for securely evaluating $\llbracket\mathrm{ReLU}(x)\rrbracket^{A}$ in additive secret sharing:
\begin{enumerate}
	\item $P_{1}$ randomly generates $r\in\mathbb{Z}_{2^{k}}$ and defines $m_{b\in\{0,1\}}:=(b\oplus \langle  msb(x)' \rangle_{1})\times\langle x\rangle_{1}-r$,
	and sends them to $P_{2}$.
	\item $P_{2}$ chooses $m_{b}$ based on $ \langle  msb(x)' \rangle_{2}$, namely, $P_{2}$ chooses $m_{0}$ if $ \langle  msb(x)' \rangle_{2}=0$, and otherwise $P_{2}$ chooses $m_{1}$. Therefore, the secret share held by $P_{2}$ is $m_{\langle  msb(x)' \rangle_{2}}= msb(x)'\times\langle x\rangle_{1}-r$, and the secret share held by $P_{1}$ is $r$.
	
	\item For the other secret share $\langle x\rangle_{2} $, $P_{2}$ acts as the sender and $P_{1}$ acts as the receiver to perform step 1) and 2) again.
\end{enumerate}

Finally, $P_{\{1,2\}}$ hold the secret shares $\langle msb(x)'\times x\rangle_{\{1,2\}}$. It is noted that, in \cite{mohassel2018aby3}, $P_{\{1,2\}}$ should re-share their shares to $P_{3}$ since they work on replicated secret sharing. Subroutine \ref{alg:a4} describes our protocol for secure ReLU function.

\begin{algorithm}[!t]
	\caption{Secure ReLU Function} 
	\label{alg:a4}
	\begin{algorithmic}[1] 
		\REQUIRE  $P_{\{1,2\}}$ hold node $v_{i}$'s $0_{th}$-layer encrypted state : $\{\llbracket \hat{\mathbf{x}}^{(0)}_{v_{i},1} \rrbracket^{A}, \cdots, \llbracket \hat{\mathbf{x}}^{(0)}_{v_{i},H}\rrbracket^{A}\}$, where $\hat{\mathbf{x}}^{(0)}_{v_{i}}=\overline{\mathbf{x}}^{(0)}_{v_{i}}\mathbf{M}^{(1)}$.
		\ENSURE $P_{\{1,2\}}$ hold node $v_{i}$'s $1_{st}$-layer encrypted state: $\{\llbracket \mathbf{x}^{(1)}_{v_{i},1} \rrbracket^{A}, \cdots, \llbracket \mathbf{x}^{(1)}_{v_{i},H}\rrbracket^{A}\}$.
		
		\FOR{$j=1$ to $H$}
		
		\STATE Securely calculating $\llbracket msb(\hat{\mathbf{x}}^{(0)}_{v_{i},j})\rrbracket^{B}$ by tailored PPA.
		
		\STATE One of $P_{\{1,2\}}$ flips its share, and then $P_{\{1,2\}}$'s shares are $\langle  msb(\hat{\mathbf{x}}^{(0)}_{v_{i},j})' \rangle_{1}$ and $\langle  msb(\hat{\mathbf{x}}^{(0)}_{v_{i},j})' \rangle_{2}$, respectively.
		
		\STATE $P_{1}$ randomly generates $r\in\mathbb{Z}_{2^{k}}$, and sends $m_{b}:=(b\oplus \langle  msb(\hat{\mathbf{x}}^{(0)}_{v_{i},j})' \rangle_{1})\times\langle \hat{\mathbf{x}}^{(0)}_{v_{i},j}\rangle_{1}-r, b\in\{0,1\}$ to $P_{2}$.
		
		\STATE$P_{2}$ chooses $m_{b}$ based on $ \langle  msb(\hat{\mathbf{x}}^{(0)}_{v_{i},j})' \rangle_{2}$.
		
		\STATE $P_{2}$ randomly generates $r'\in\mathbb{Z}_{2^{k}}$, and sends $m_{b}':=(b\oplus \langle  msb(\hat{\mathbf{x}}^{(0)}_{v_{i},j})' \rangle_{2})\times\langle \hat{\mathbf{x}}^{(0)}_{v_{i},j}\rangle_{2}-r', b\in\{0,1\}$ to $P_{1}$.
		
		\STATE$P_{1}$ chooses $m_{b}'$ based on $ \langle  msb(\hat{\mathbf{x}}^{(0)}_{v_{i},j})' \rangle_{1}$.
		
		\STATE Finally, $P_{1}$ holds $msb(\hat{\mathbf{x}}^{(0)}_{v_{i},j})'\times\langle \hat{\mathbf{x}}^{(0)}_{v_{i},j}\rangle_{2}-r'+r$ and $P_{2}$ holds $msb(\hat{\mathbf{x}}^{(0)}_{v_{i},j})'\times\langle \hat{\mathbf{x}}^{(0)}_{v_{i},j}\rangle_{1}-r+r'$.
		\ENDFOR
	\end{algorithmic}
\end{algorithm}

\subsubsection{Secure Softmax Function}
\label{sec:Softmax}
GCN usually considers a multi-classification task, which requires the Softmax function (i.e., Eq. \ref{eq:Softmax}) to normalize the probabilities of inference results. Therefore, we need a protocol to securely compute the Softmax function.

First, to avoid error from calculating the exponential function on very large or very small values, a frequently-used method is to calculate the Softmax function on $\mathbf{x}-max(\mathbf{x})$. When calculating $max(\mathbf{x})$ in the secret sharing domain, to reduce the overhead, we can use the binary-tree form, e.g., $max(max(\llbracket x_{1}\rrbracket^{A},\llbracket x_{2}\rrbracket^{A}),max(\llbracket x_{3}\rrbracket^{A},\llbracket x_{4}\rrbracket^{A}))$, which requires $\mathrm{log}C$ rounds comparison and $C$ is the number of classifications. We can directly use the secure ReLU function introduced above to perform $max()$:
\begin{equation}
	max(\llbracket x_{1}\rrbracket^{A},\llbracket x_{2}\rrbracket^{A})= \mathrm{ReLU}(\llbracket x_{1}\rrbracket^{A}-\llbracket x_{2}\rrbracket^{A})+\llbracket x_{2}\rrbracket^{A}.\notag
\end{equation}

After that, the cloud servers should compute $\llbracket e^{ x}\rrbracket^{A}$. Since $\llbracket e^{ x}\rrbracket^{A}$ is not naturally supported in the secret sharing domain, we first approximate $e^{ x}$ using its limit characterization \cite{knott2020crypten}:
\begin{equation}
\label{eq:exp}
e^{x}\approx (1+\frac{ x}{2^{n}})^{2^{n}}.
\end{equation}
However, the approximation is inefficient if the cloud servers serially calculate the multiplication, which will require to calculate $2^{n}$ multiplications in $2^{n}$ rounds communication. Our solution is to calculate the approximation by the binary-tree form. More specifically, the core of Eq. \ref{eq:exp} is to calculate $( \llbracket x\rrbracket^{A})^{2^{n}}$, thus $P_{\{1,2\}}$ first calculate $( \llbracket x\rrbracket^{A})^{2}$ in one round, and then set $\llbracket y\rrbracket^{A}=(\llbracket x\rrbracket^{A})^{2}$ followed by calculating $(\llbracket y\rrbracket^{A})^{2}$ in one round. Therefore, $P_{\{1,2\}}$ can calculate $(\llbracket x\rrbracket^{A})^{2^{n}}$ in $\mathrm{log}~2^{n}=n$ rounds.  Subroutine \ref{alg:a5} describes our protocol for secure Softmax function.

\begin{algorithm}[!t]
	\caption{Secure Softmax Function} 
	\label{alg:a5}
	\begin{algorithmic}[1] 
	    \REQUIRE  $P_{\{1,2\}}$ hold node $v_{i}$'s $1_{th}$-layer encrypted state : $\{\llbracket \hat{\mathbf{x}}^{(1)}_{v_{i},1} \rrbracket^{A}, \cdots, \llbracket \hat{\mathbf{x}}^{(1)}_{v_{i},C}\rrbracket^{A}\}$, where $\hat{\mathbf{x}}^{(1)}_{v_{i}}=\overline{\mathbf{x}}^{(1)}_{v_{i}}\mathbf{M}^{(2)}$.
		\ENSURE $P_{\{1,2\}}$ hold node $v_{i}$'s $2_{nd}$-layer encrypted state: $\{\llbracket \mathbf{x}^{(2)}_{v_{i},1} \rrbracket^{A}, \cdots, \llbracket \mathbf{x}^{(2)}_{v_{i},C}\rrbracket^{A}\}$.
		
		\STATE Calculate $\llbracket Q \rrbracket^{A}=max\{\llbracket \hat{\mathbf{x}}^{(1)}_{v_{i},j} \rrbracket^{A} \}_{j\in[1,C]}$ by $ \mathrm{ ReLU}()$.
		
		\STATE Locally calculate $\llbracket\hat{\mathbf{x}}'^{(1)}_{v_{i},j}\rrbracket^{A}=\llbracket\hat{\mathbf{x}}^{(1)}_{v_{i},j}\rrbracket^{A}-\llbracket Q \rrbracket^{A},j\in[1,C]$.
		
		 //$P_{\{1,2\}}$ calculate the approximate $\llbracket e^{ \hat{\mathbf{x}}'^{(1)}_{v_{i},j} }\rrbracket^{A}$:
		\FOR{$j=1$ to $C$}
		\STATE $\llbracket y_{0} \rrbracket^{A}=1+\frac{\llbracket \hat{\mathbf{x}}'^{(1)}_{v_{i},j} \rrbracket^{A}}{2^{\mathcal{N}}}$.
		\FOR{$n=0$ to $\mathcal{N}$}
		\STATE$\llbracket y_{n+1} \rrbracket^{A}=\llbracket y_{n} \rrbracket^{A}\times \llbracket y_{n} \rrbracket^{A}$.
		\ENDFOR
		\ENDFOR
		\STATE $P_{\{1,2\}}$ first locally calculate $\llbracket S\rrbracket^{A}= \sum_{j\in[1,C]}\llbracket e^{ \hat{\mathbf{x}}'^{(1)}_{v_{i},j} }\rrbracket^{A}$, and then calculate the approximate $\llbracket\frac{1}{S}\rrbracket^{A}$ by Eq. \ref{recip}.
		
		//$P_{\{1,2\}}$ calculate the $2_{nd}$-layer encrypted state:
		\FOR{$j=1$ to $C$}
		\STATE $\llbracket \mathbf{x}^{(2)}_{v_{i},j} \rrbracket^{A}=\llbracket e^{ \overline{\mathbf{x}}'^{(1)}_{v_{i},j} }\rrbracket^{A}\times \llbracket\frac{1}{S }\rrbracket^{A}$.
		\ENDFOR
	\end{algorithmic}
\end{algorithm}
 \vspace{-13pt}
\subsection{Secure Model Convergence Evaluation}
\label{sec:stop}

So far we have presented our solution for securely realizing the forward propagation process as given in Eq. \ref{eq:GCN} in the secret sharing domain. We now show how to securely evaluate the convergence of the model training process. 

We note that prior works (e.g., \cite{mohassel2018aby3,wagh2020falcon,tan2021cryptgpu}) on secure CNN training generally terminate the training process at a specified number of epochs. However, the convergence of the training process is unpredictable, which can depend on various factors such as the training data set, the learning parameter setting, and random factors in the nature of model training. A fixed number of epochs without considering the property of models may easily lead to overfitting or underfitting \cite{AalstRVDKG10}. Therefore, instead of specifying a certain number of epochs, it is much more desirable to directly evaluate the model convergence in a secure manner. 

Our solution is to calculate the encrypted cross-entropy loss and then calculate the difference in the encrypted cross-entropy loss between two adjacency epochs. If the difference is smaller than a public threshold $\alpha$ and lasts for a window size, the cloud servers $P_{\{1,2\}}$ will conclude that the model is convergent and will terminate the training. 
From the computation, $P_{\{1,2\}}$ know nothing except the \emph{necessary} fact about whether the difference in the cross-entropy loss between two adjacency epochs is less than $\alpha$. 

A new challenge arises, namely, how to calculate the cross-entropy loss in the secret sharing domain. In Eq. \ref{eq:loss}, the \textit{natural logarithm} is not naturally supported in the secret sharing domain, and requires a tailored protocol. Inspired by \cite{knott2020crypten}, we approximate $\mathrm{ln}~x$ by:
\begin{equation}
	\label{log}
	y_{n+1}=y_{n}-\sum_{k\in[1,\mathcal{K}]}\frac{1}{k}(1-xe^{-y_{n}})^{k},
\end{equation}
which will converge to $  y_{n}\approx \mathrm{ln} ~x$. The initial value can be set as $y_{0}=\frac{x}{120}-20e^{-2x-1}+3$ \cite{knott2020crypten}. Obviously, both subtraction and multiplication are naturally supported in secret sharing domain, and $\llbracket e^{-2 x-1}\rrbracket^{A}$ can be calculated by  Eq. \ref{eq:exp}.

After obtaining the encrypted loss $\llbracket  \mathcal{L}^{j} \rrbracket^{A}$ and $\llbracket  \mathcal{L}^{j+1} \rrbracket^{A}$ of two adjacent epochs, $P_{\{1,2\}}$ first calculate the absolute value of their difference:
\begin{align}\notag
|\llbracket  \mathcal{L}^{j+1} \rrbracket^{A}-\llbracket  \mathcal{L}^{j} \rrbracket^{A}|=&\mathrm{ReLU}(\llbracket  \mathcal{L}^{j+1} \rrbracket^{A}-\llbracket  \mathcal{L}^{j} \rrbracket^{A})\\\notag
+&\mathrm{ReLU}(\llbracket  \mathcal{L}^{j} \rrbracket^{A}-\llbracket  \mathcal{L}^{j+1} \rrbracket^{A}).
\end{align}
Then, the model convergence flag is calculated by $\llbracket msb(\alpha-|\llbracket  \mathcal{L}^{j} \rrbracket^{A}-\llbracket  \mathcal{L}^{j+1} \rrbracket^{A}|)\rrbracket^{B}$. $P_{\{1,2\}}$ open the flag to each other, and then decide whether to terminate the training. Subroutine \ref{alg:a6} describes our protocol for secure model convergence evaluation.

\vspace{-10pt}
\subsection{Putting Things Together}
\label{sec:train-infer}

\noindent\textbf{Secure training.} When training the GCN model, the cloud servers first securely normalize all nodes' initial features through \emph{secure feature normalization}. 
After that, the cloud servers securely perform the forward propagation (Eq. \ref{eq:GCN}) through \emph{secure neighboring states aggregation}, and \emph{secure activation functions} for each labeled node $v_{i}\in \mathcal{T}$ to obtain the inference results $\mathbf{Z}_{v_{i}, j}, j\in[1,C]$. 
Subsequently, the cloud servers securely calculate the average cross-entropy loss $\mathcal{L}$ between each labeled node's inference result $\mathbf{Z}_{v_{i}}$ and its true label $\mathbf{T}_{v_{i}}$ and then securely evaluate the model convergence.

\begin{algorithm}[!t]
	\caption{Secure Model Convergence Evaluation} 
	\label{alg:a6}
	\begin{algorithmic}[1] 
		\REQUIRE  $P_{\{1,2\}}$ hold all labeled nodes' encrypted inference results $\{\llbracket \mathbf{Z}_{v_{i}} \rrbracket^{A}\}_{v_{i}\in \mathcal{T}}$ and encrypted labels $\{\llbracket \mathbf{T}_{v_{i}} \rrbracket^{A}\}_{v_{i}\in \mathcal{T}}$, the public threshold $\alpha$, the public window size $\beta$ and the public maximum number of epochs $\gamma$.
		\ENSURE Nothing.
		\FOR{$e=1$ to $\gamma$}
		\STATE $\llbracket  \mathcal{L}^{e} \rrbracket^{A}=0$.
		
	    // $P_{\{1,2\}}$ calculate the approximate $\llbracket  \mathrm{ln}~ \mathbf{Z}_{v_{i},j}\rrbracket^{A}$:
	    
	    \FOR{each $v_{i}\in \mathcal{T}, j\in[1,C]$}
	    
	    \STATE $\llbracket   y_{0}\rrbracket^{A}=\frac{\llbracket 
	     \mathbf{Z}_{v_{i},j}\rrbracket^{A}}{120}-20\times\llbracket e^{-2\times   \mathbf{Z}_{v_{i},j}-1}\rrbracket^{A}+3$.
     
	    \FOR{$n=0$ to $\mathcal{N}$}
	    
	    \STATE $\llbracket   y_{n+1}\rrbracket^{A}=\llbracket   y_{n}\rrbracket^{A}-\sum_{k=1}^{\mathcal{K}}\frac{1}{k}\times(1-\llbracket  \mathbf{Z}_{v_{i},j}\rrbracket^{A}\times\llbracket e^{-   y_{n}}\rrbracket^{A})^{k}$.
	    
	    \ENDFOR
		\ENDFOR
		
		// $P_{\{1,2\}}$ calculate the cross-entropy loss $\llbracket  \mathcal{L}^{e} \rrbracket^{A}$:
		\FOR{each $v_{i}\in \mathcal{T}, j\in[1,C]$}
		\STATE $\llbracket  \mathcal{L}^{e} \rrbracket^{A}-=\llbracket \mathbf{T}_{v_{i},j} \rrbracket^{A}\times \llbracket  \mathrm{ln}~ \mathbf{Z}_{v_{i},j}\rrbracket^{A}$.
		\ENDFOR
		
		// $P_{\{1,2\}}$ determine whether to stop the training:
		\STATE $|\llbracket  \mathcal{L}^{e} \rrbracket^{A}-\llbracket  \mathcal{L}^{e-1} \rrbracket^{A}|=\mathrm{ReLU}(\llbracket  \mathcal{L}^{e} \rrbracket^{A}-\llbracket  \mathcal{L}^{e-1} \rrbracket^{A})+\mathrm{ReLU}(\llbracket  \mathcal{L}^{e-1} \rrbracket^{A}-\llbracket  \mathcal{L}^{e} \rrbracket^{A})$.
		\STATE 	flag=$\llbracket msb(\alpha-|\llbracket  \mathcal{L}^{j} \rrbracket^{A}-\llbracket  \mathcal{L}^{j+1} \rrbracket^{A}|)\rrbracket^{B}$.
		\STATE $stop~=(flag~==~1)~?~0:~(stop~+~1)$.
		
		\STATE \textbf{if} ($stop \geqslant \beta$) \textbf{then} terminating the training process.
		\ENDFOR
	\end{algorithmic}
\end{algorithm}

If convergence is not yet achieved, the cloud servers perform backward propagation to calculate each trainable weight's gradient $\frac{\partial\mathcal{L}}{\partial \mathbf{M}_{i,j}}$ followed by updating each weight using its gradients. Based on the chain rule \cite{he2012geometrical}, if the cloud servers can calculate the derivatives of all non-linear functions, they can calculate the complete derivative of Eq. \ref{eq:GCN}. In Eq. \ref{eq:GCN}, the first non-linear function is the cross-entropy loss function, and its derivative is:
\begin{equation}\notag
		\frac{\partial \mathcal{L}}{\partial \mathbf{Z}_{v_{i},j}}=-\frac{\mathbf{T}_{v_{i},j}}{\mathbf{Z}_{v_{i},j}}, v_{i}\in \mathcal{T},j\in[1,C],
\end{equation}
where the division can be securely calculated by using the design in Section \ref{sec:initia}. The second non-linear function is the $\mathrm{Softmax}$ function, and its derivative is:
\begin{align}\notag
	\frac{\partial z_{j}}{ \partial x_{i}}= \begin{cases}
		z_{j}(1-z_{j}) & \text{if } i=j,\\
		-z_{j}z_{i} & \text{if }i \ne j,
	\end{cases}
\end{align}
where $z_{j}=Softmax(x_{j})$, which can be securely calculated by using Subroutine \ref{alg:a5} in Section \ref{sec:activation}. The third non-linear function is the $ \mathrm{ReLU}$ function, and its derivative is: 
\begin{align}\notag
	\frac{\partial \mathrm{ReLU}(x)}{\partial x}=\begin{cases}
		0 & \text{if } x < 0,\\
		1 & \text{if } x > 0, 
	\end{cases}
\end{align}
which can be securely calculated by using the tailored PPA in Section \ref{sec:ReLU}. 
So this is the whole process of secure training in our system.

\noindent\textbf{Secure inference.} 
Secure inference for an unlabeled node corresponds to a forward propagation through the trained GCN model in the secret sharing domain.
In particular, the data owner provides the cloud servers with the ID of the unlabeled node. Upon receiving the ID, the cloud servers securely conduct the forward propagation process (i.e., Eq. \ref{eq:GCN}) in the secret sharing domain, and output the encrypted inference result about its label, which is then sent to the data owner for reconstruction.

\vspace{-10pt}
\section{Security Analysis}
\label{sec:security-analysis}

We follow the standard ideal/real world paradigm to analyze the security of {\secgnn}.
In the ideal/real world paradigm, a protocol is secure if the view of the corrupted party during the real execution of a protocol can be generated by a simulator given only the party’s input and legitimate output, which can be defined as follows:

\begin{definition}
	\label{def1}
	Let $P_{\{1,2,3\}}$ engage in a protocol $\pi$ which computes function $f:(\{0,1\}^{*})^{3}\rightarrow (\{0,1\}^{*})^{3}$. $P_{i}$'s view during the execution of protocol $\pi$  on inputs $\mathbf{x}$, denoted as $\mathrm{View}_{i}^{\pi}(\mathbf{x})$, consists of its input $\mathbf{in}_{i}$, its internal random values $\mathbf{r}_{i}$ and the messages $\mathbf{m}_{i}$ received during the execution. We say that $\pi$ computes $f$ with security in the semi-honest and non-colluding setting, if there exists a probabilistic polynomial time simulator $Sim$ such that for each $P_{i}$: $\mathrm{Sim}(\mathbf{in}_{i},f_{i}(\mathbf{x})) \mathop  \approx \limits \mathrm{View}_{i}^{\pi}(\mathbf{x})$.
\end{definition}

Recall that {\secgnn} consists of several secure sub-protocols: 1) secure division $\mathtt{secDIV}$; 2) secure array access $\mathtt{secACCESS}$; 3) secure square root $\mathtt{secROOT}$; 4) secure $\mathrm{ReLU}$ function $\mathtt{secRELU}$; 5) secure Softmax function $\mathtt{secSoftmax}$; 6) secure natural logarithm $\mathtt{secLOG}$. We use $Sim^{P_{i}}_{\mathtt{X}}$ to denote the simulator which can generate $P_{i}$'s view in sub-protocol $\mathtt{X}$ on corresponding input and output.

\begin{theorem}
	\label{theo1}
	Our {\secgnn} is secure according to Definition \ref{def1}.
\end{theorem}

\begin{proof}

It is noted that the inputs and outputs of each sub-protocol are secret shares, with each sub-protocol being invoked in order as per the processing pipeline.
If the simulator for each sub-protocol exists, then our complete protocol is secure \cite{curran2019procsa}. 
It is easy to see that the simulators $Sim_{\mathtt{X}}^{P_{i\in\{1,2,3\}}}$ ($\mathtt{X}\in\{\mathtt{secDIV}, \mathtt{secROOT}, \mathtt{secSoftmax} ,\mathtt{secLOG}$) must exist, because they are all calculated through approximations which are realized via basic operations (i.e., addition and multiplication) in the secret sharing domain.
Therefore, {\secgnn} is secure if the simulators for the remaining sub-protocols exist, i.e., $\mathtt{secACCESS}$ in Section \ref{sec:access} and $\mathtt{secRELU}$ in Section \ref{sec:ReLU}.
The existence of these simulators is given in \textbf{\emph{Theorem}} \ref{theo2} and \textbf{\emph{Theorem}} \ref{theo3}. 
\end{proof}

\begin{theorem}
	\label{theo2}
	The protocol $\mathtt{secACCESS}$ for secure neighboring states access is secure according to Definition \ref{def1}.
\end{theorem}

\begin{proof}

We consider the simulator of $P_{1}$, $P_{2}$ and $P_{3}$ in turn.
\begin{itemize}
\item 	$Sim^{P_{1}}_{\mathtt{secACCESS}}$: The simulator is simple since $P_{1}$ receives nothing in the real execution. Therefore, it is clear that the simulated view is identical to the real view.

\item $Sim^{P_{2}}_{\mathtt{secACCESS}}$: To analyze $P_{2}$'s view, we see that $P_{2}$ has $k_{1}, k_{2}$, and shares $ \langle I \rangle_{2}$,  $\langle  \mathbf{a}\rangle_{2}$ at the beginning, and later receives new shares $ \langle \mathbf{a}''\rangle_{2}$ and $\langle h \rangle_{1}$ in step 1). In the simulated view, $P_{2}$ receives random values in step 1).  Therefore, we need to prove that $ \langle \mathbf{a}''\rangle_{2}$ and $\langle h \rangle_{1}$ are uniformly random in the view of  $P_{2}$.
\begin{itemize}
\item $ \langle \mathbf{a}''\rangle_{2}$ are uniformly random in $P_{2}$'s view: Firstly, $\langle  \mathbf{a}'' \rangle_{2}=(\langle  \mathbf{a} \rangle _{1}\circlearrowright r^{1} +\mathbf{c}^{1})\circlearrowright r^{3}$
and $\mathbf{c}^{1}[j]=\mathbb{F}(k_{1},j)-\mathbb{F}(k_{3},j), j\in[1,m]$.
Though $P_{2}$ has $k_{1}$, it does not have $k_{3}$, thus $\mathbb{F}(k_{3},j)$ is uniformly random in $P_{2}$'s view. It implies that $\mathbf{c}^{1}[j]$ is also uniformly random in $P_{2}$'s view since $\mathbb{F}(k_{3},j)$ is independent of $\mathbb{F}(k_{1},j)$ used in the generation of $\mathbf{c}^{1}[j]$ \cite{araki2016high}. Similarly, the array $\langle  \mathbf{a}'' \rangle_{2}$ is uniformly random in $P_{2}$'s view since $\mathbf{c}^{1}$ is independent of $\langle  \mathbf{a}\rangle_{1}$ used in the generation of $\langle  \mathbf{a}'' \rangle_{2}$. Therefore, the distribution over the real $\langle  \mathbf{a}'' \rangle_{2}$ received by $P_{2}$ in the protocol execution and over the simulated $\langle  \mathbf{a}'' \rangle_{2}$ generated by the simulator is identically distributed. 

\item $\langle h \rangle_{1}$ is uniformly random in $P_{2}$'s view: In a similar way, $ \langle h\rangle _{1}=\langle I\rangle_{1}+r^{1}+r^{3}$, where $r^{1}=\mathbb{F}(k_{1},j)$ and $r^{3}=\mathbb{F}(k_{3},j)$. Though $P_{2}$ has $k_{1}$, it does not have $k_{3}$, thus $r^{3}$ is uniformly random in $P_{2}$'s view, furthermore, $ \langle h\rangle _{1}$ is uniformly random in $P_{2}$'s view. Therefore, the distribution over $ \langle h\rangle _{1}$ received by $P_{2}$ in the protocol execution and over the $ \langle h\rangle _{1}$ generated by the simulator is identically distributed.
\end{itemize}

\item  $Sim^{P_{3}}_{\mathtt{secACCESS}}$: To analyze $P_{3}$'s view, we see that $P_{3}$ has $k_{2}, k_{3}$ at the beginning, and later receives share $\langle   \mathbf{a}' \rangle_{2}$ and $h$ in step 3). It is noted that the proof of $Sim^{P_{3}}_{\mathtt{secACCESS}}$ is similar to the proof of $Sim^{P_{2}}_{\mathtt{secACCESS}}$ since $P_{3}$ and $P_{2}$ receive similar messages during the protocol execution, thus we omit the proof of $Sim^{P_{3}}_{\mathtt{secACCESS}}$.
\end{itemize}
\end{proof}

\begin{theorem}
	\label{theo3}
	The protocol $\mathtt{secRELU}$ for the $\mathrm{ReLU}$ function is secure according to Definition \ref{def1}.
\end{theorem}

\begin{proof}
	Obviously, the $\llbracket msb(\llbracket x\rrbracket^{A})\rrbracket^{B}$ function is secure since the tailored PPA consists of basic AND and XOR gates, so we only prove that $\llbracket msb(x)\rrbracket^{B}\times \llbracket x\rrbracket^{A}$ function is secure. In the case of $P_{1}$ acting as the sender and $P_{2}$ acting as the receiver, we consider the simulator of $P_{1}$, $P_{2}$ and $P_{3}$  in turn.
	\begin{itemize}
		\item $Sim^{P_{1}}_{\mathtt{secRELU}}$: The simulator is simple since $P_{1}$ receives nothing in the real execution. Therefore, it is clear that the simulated view is identical to the real view.
		
		\item $Sim^{P_{2}}_{\mathtt{secRELU}}$: To analyze $P_{2}$'s view, we see that $P_{2}$ has $\langle msb(x)'\rangle_{2}$ and $\langle x\rangle_{2}$ at the beginning, and later receives messages $m_{b}:=(b\oplus \langle  msb(x)' \rangle_{1})\times \langle x\rangle_{1}-r, b\in\{0,1\}$. In the simulated view, $P_{2}$ receives two random values. Therefore, we need to prove that $m_{\{1,2\}}$ are uniformly random in the view of  $P_{2}$. Obviously, the above claim is valid, because $r$ is uniformly random in $P_{2}$'s view, which implies that $m_{\{1,2\}}$ are also uniformly random in $P_{2}$'s view since $r$ is independent of other values used in the generation of $m_{\{1,2\}}$. Therefore, the distribution over the real $m_{\{1,2\}}$ received by $P_{2}$ in the protocol execution and over the simulated $m_{\{1,2\}}$ generated by the simulator is identically distributed. 
		
		\item $Sim^{P_{3}}_{\mathtt{secRELU}}$: The simulator is simple since $P_{3}$ does not participate in the protocol and receives nothing in the real execution. Therefore, it is clear that the simulated view is identical to the real view.
	\end{itemize}
	Similarly, in the case of $P_{2}$ acting as the sender and $P_{1}$ acting as the receiver, the protocol is also secure. 
\end{proof}

\noindent \textbf{Discussion.}
As the first research endeavor towards privacy-preserving training and inference of GNNs outsourced to the cloud, the current design of {\main} only considers the commonly assumed non-colluding and semi-honest threat model, where the three cloud servers $P_{\{1,2,3\}}$ will not collaboratively launch inference attacks, e.g., model inversion attack \cite{zhang2022model}.
On another hand, we are aware that there exist effective mechanisms for bounding information leakage even if $P_{\{1,2,3\}}$ collude with each other, which can also be smoothly integrated into {\main} for security enhancement.
Specifically, we observe that local differential privacy (LDP) \cite{kasiviswanathan2011can} and dummy edges padding are promising techniques, of which the blueprint is as follows.
It is noted that the private information in the graph-structured data that needs to be protected is the node features and labels and the edges between nodes.
Firstly, before encrypting the graph-structured data,  the data owner perturbs the node features and labels by the LDP-based obfuscation mechanism \cite{sajadmanesh2021locally}, which is specifically designed for GNNs.
Secondly, the data owner adds dummy edges with random weights between some pairs of unconnected nodes in the graph-structured data to obfuscate the existence and weights of edges.  
Finally, the data owner encrypts
the graph-structured data after obfuscation by the encryption method introduced in Section \ref{subsec:sec-input-preparation}.
Since the node features and labels and the edges between nodes in the graph-structured data are obfuscated, $P_{\{1,2,3\}}$ cannot learn the accurate original graph-structured data even if they collude with each other.
So the above is the blueprint for prevent $P_{\{1,2,3\}}$ from colluding with each other to launch inference attacks in {\main}, for which it is important to explore how to make the decreased accuracy of the trained GNN model (a natural trade-off) as small as possible upon concrete realizations.

\section{Experiments}
\label{sec:experiments}

\subsection{Setup}

The implementation is written in C++ using the standard library. 
All experiments are performed on a workstation with Intel Core i7-10700K and 64GB RAM running Ubuntu 20.04.2 LTS. Consistent with prior art \cite{riazi2019xonn,tan2021cryptgpu}, we consider a Local Area Network (LAN) environment with a network bandwidth of 625MB/s and an average latency of 0.22 ms. For all experiments, we split our computation and communication into data-dependent online phase and data-independent offline phase, and report the end-to-end protocol execution time and the total communication traffic.
Our implementation is available at \url{https://github.com/songleiW/SecGNN}.

\noindent\textbf{Graph datasets.} We use three graph datasets commonly used in GCN: Citeseer\footnote{\url{https://linqs-data.soe.ucsc.edu/public/lbc/citeseer.tgz}}, Cora\footnote{\url{https://linqs-data.soe.ucsc.edu/public/lbc/cora.tgz}} and Pubmed\footnote{\url{https://linqs-data.soe.ucsc.edu/public/Pubmed-Diabetes.tgz}} in our experiments. Their statistics are summarized in Table \ref{Tab:dataset}.

\begin{table}
	\normalsize 
	\centering
	\caption{Dataset Statistics}
	\label{Tab:dataset}
	\begin{tabular*}{\hsize}{@{}@{\extracolsep{\fill}}cccccc@{}}
		\toprule
		Dataset &Nodes&Edges&$d_{max}$&Classes&Features \\\midrule
		Citeseer&3,327&4,732&100&6&3,703\\ \hline
		Cora&2,708&5,429&169&7&1,433\\ \hline
		Pubmed& 19,717&44,338&171&3&500\\ 
		\bottomrule
	\end{tabular*}
\vspace{-10pt}
\end{table}

\noindent\textbf{Model hyperparameters.}  Similar to \cite{GCN}, we use the two-layer GCN described in Eq. \ref{eq:GCN}. For training, we use 40 labeled samples per class but use feature vectors of all nodes. We perform batch gradient descent using the full training set for each epoch. The learning rate is 0.2 and the size of the hidden layer is 16. Early stopping with a window size 5 and public threshold 0.02. We randomly initialize model parameters by the uniform distribution $\mathbf{M}^{(0)}\sim(\frac{-1}{\sqrt{E}}, \frac{1}{\sqrt{E}})$, 
where $E$ is the number of neurons. We use the same hyperparameters in plaintext and {\secgnn}.

\noindent \textbf{Protocol instantiation.} We instantiate the sub-protocols in Section \ref{sec:main} using the following parameter settings. Machine learning algorithms usually perform on real numbers, while the additive secret sharing is restricted to computations over integers. Following previous works \cite{mohassel2018aby3,tan2021cryptgpu }, we use a fixed-point encoding of real numbers in our secure protocols. Specifically, for a real number $x$, we consider a fixed-point encoding with $t$ bits of precision: $\lfloor x\cdot 2^{t} \rceil$. Note that when multiplying two fixed-point encoding numbers, since both of them are multiplied by $2^{t}$, the two parties additionally need to rescale the product scaled by $2^{2t}$, where we use the truncation technique from \cite{mohassel2017secureml}. In our experiments, we consider the ring $\mathbb{Z}_{2^{64}}$ with $t = 15$ bits of precision. The number of iterations of Eq. \ref{recip} is set to 13, Eq. \ref{eq:exp} is set to 8, Eq. \ref{eq:square_root} is set to 18, Eq. \ref{log} is set to 3 and $k$ is set to 8.
\subsection{Evaluation on Secure GNN Training}

\noindent\textbf{Cross-entropy loss.} We first compare the cross-entropy loss between {\secgnn} and plaintext training. The results are summarized in Fig. \ref{fig:entropy}. It is observed that the cross-entropy loss of {\secgnn} is slightly higher than that of plaintext, but they exhibit consistent behavior.
Meanwhile, it is revealed that the training processes of {\secgnn} and plaintext terminate at the same number of epochs, which demonstrates that {\secgnn}, with security assurance, does not adversely affect the convergence of the training process.
This, in turn, also validates the effectiveness of our secure model convergence evaluation protocol in Section \ref{sec:stop}.

\begin{figure}[t]
	\begin{minipage}[t]{0.32\linewidth}
		\centering{\includegraphics[width=\linewidth]{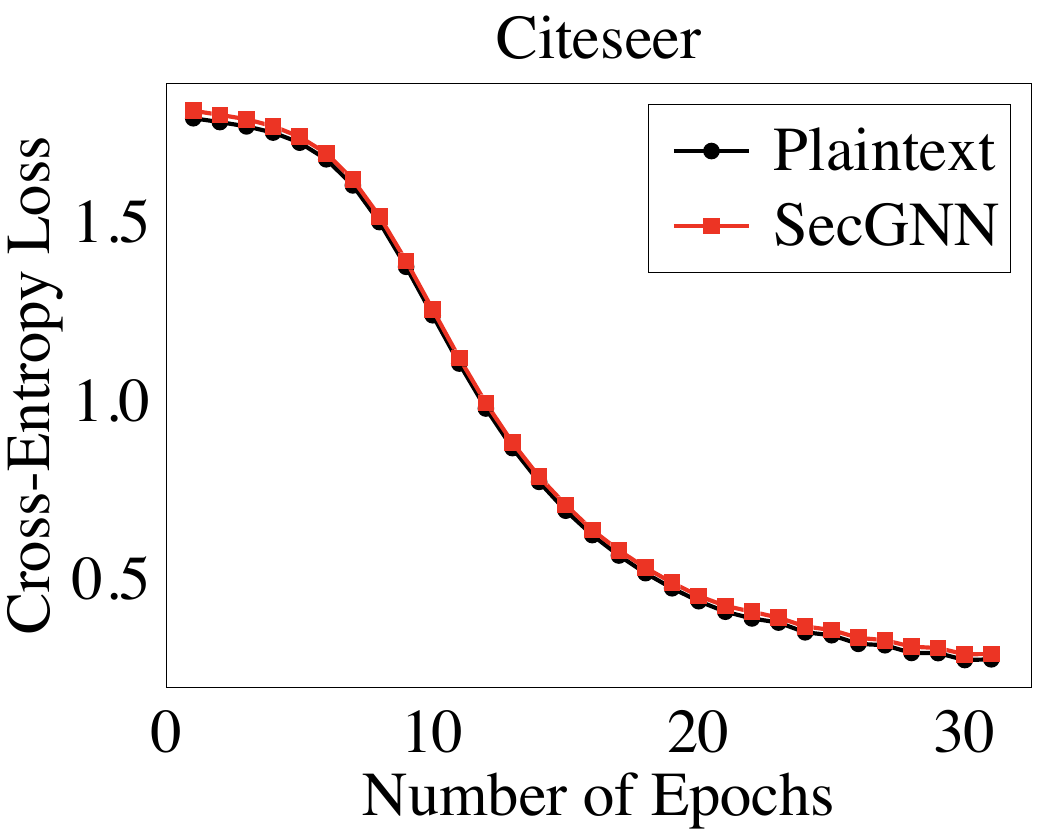}}
	\end{minipage}
	\begin{minipage}[t]{0.32\linewidth}
		\centering{\includegraphics[width=\linewidth]{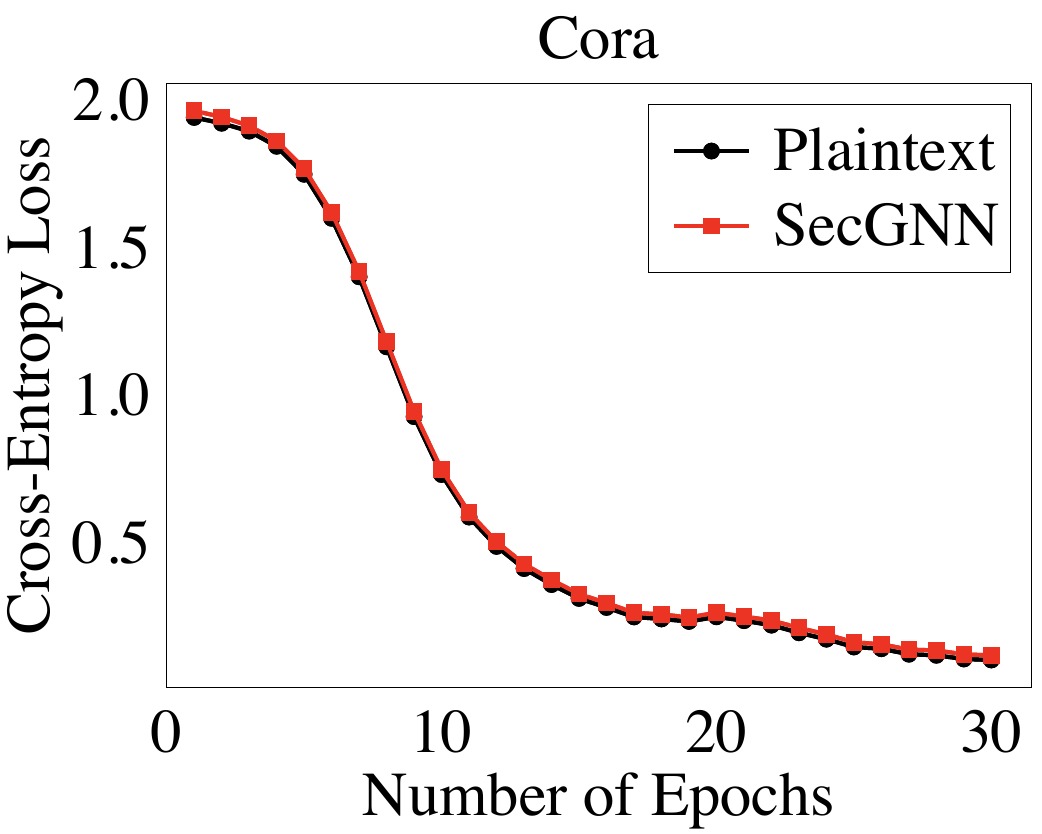}}
	\end{minipage}
	\begin{minipage}[t]{0.32\linewidth}
		\centering{\includegraphics[width=\linewidth]{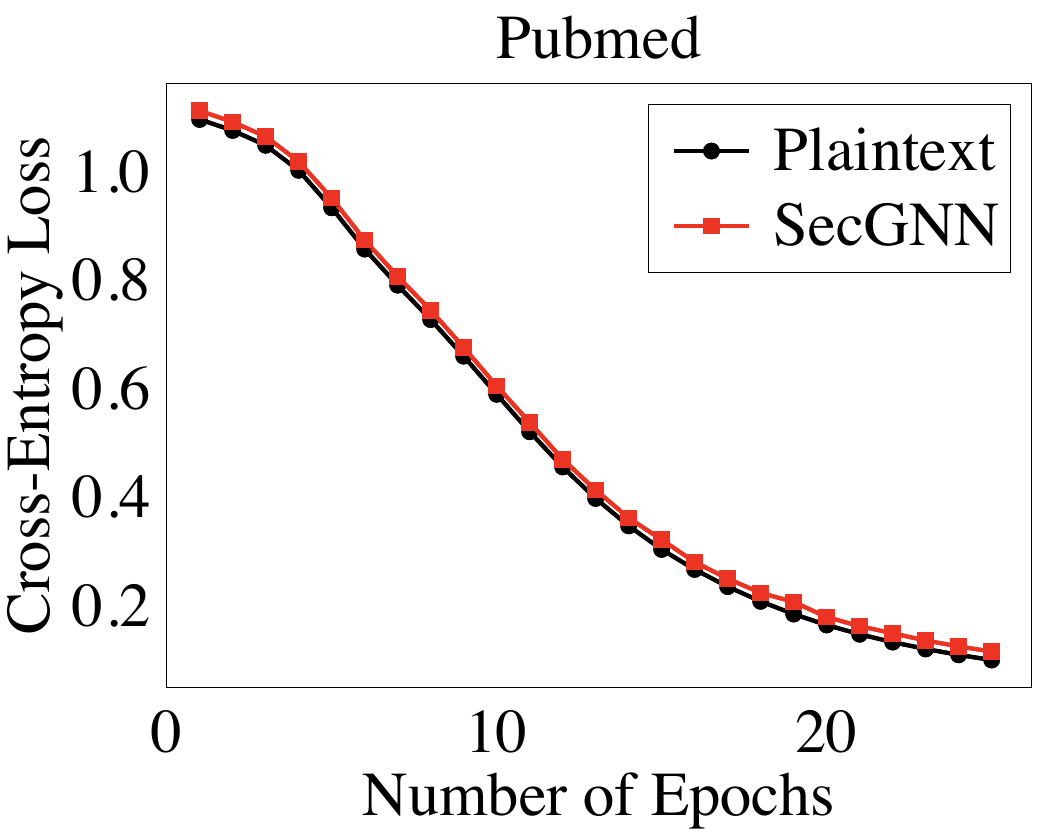}}
	\end{minipage}
	\caption{Evolution of the cross-entropy loss in SecGNN and plaintext, with varying number of epochs over different datasets.}
\label{fig:entropy}
\vspace{-10pt}
\end{figure}

\noindent\textbf{Validation set accuracy.} In addition to comparing the evolution of the cross-entropy loss, we first evaluate and compare the validation set (500 samples excluding the training samples) accuracy between {\secgnn} and plaintext. The results are summarized in Fig. \ref{fig:Valida}. It can be seen that although the difference in the validation set accuracy between {\secgnn} and plaintext is obvious at the very beginning, the difference rapidly decreases as the number of epochs grows and eventually vanishes.

\begin{figure}[t]
	\begin{minipage}[t]{0.32\linewidth}
		\centering{\includegraphics[width=\linewidth]{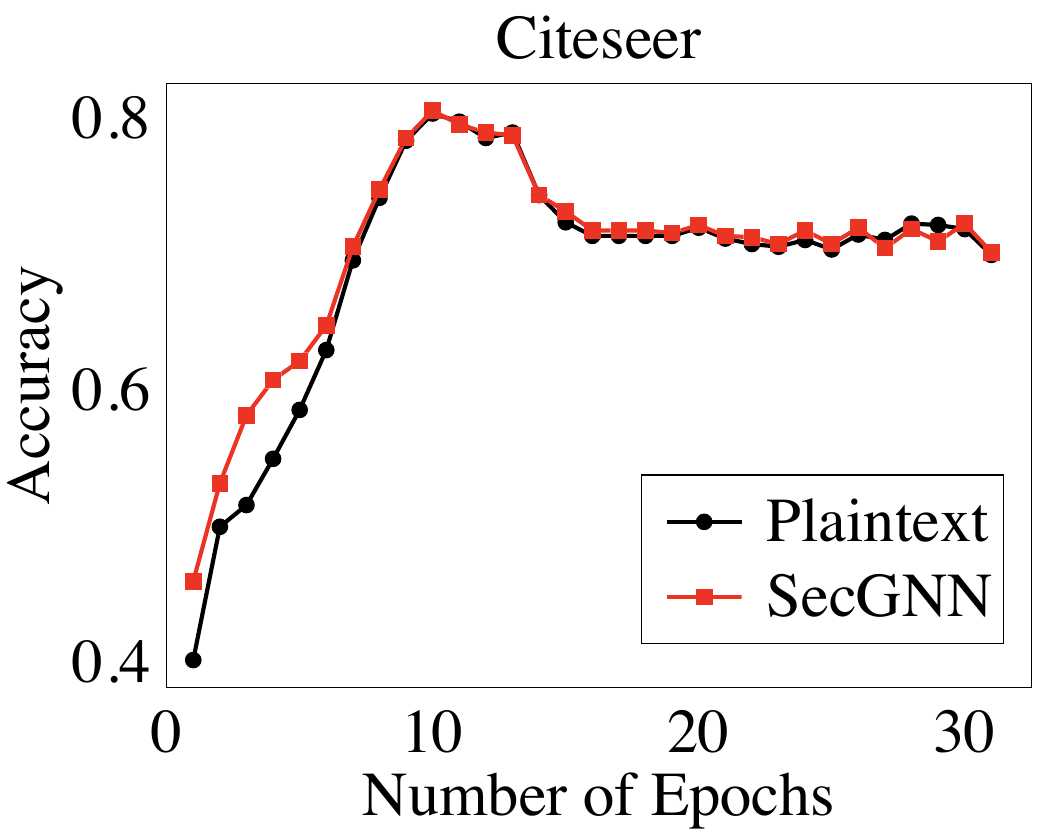}}
	\end{minipage}
	\begin{minipage}[t]{0.32\linewidth}
		\centering{\includegraphics[width=\linewidth]{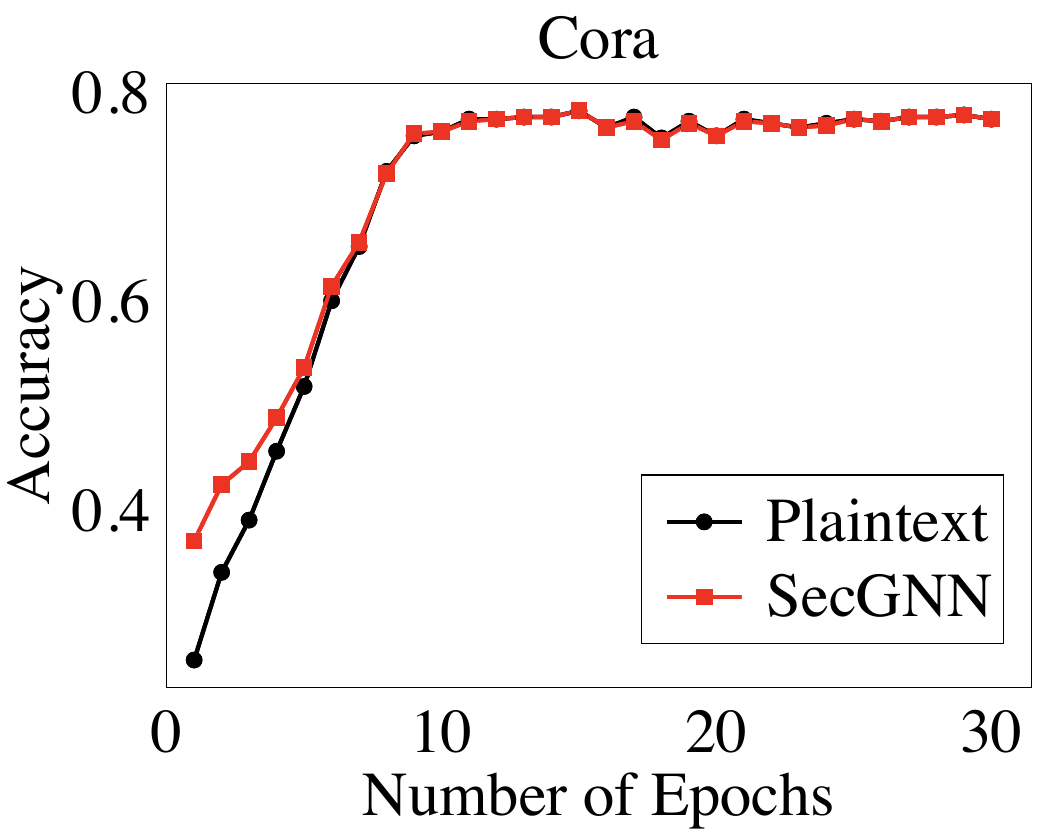}}
	\end{minipage}
	\begin{minipage}[t]{0.32\linewidth}
		\centering{\includegraphics[width=\linewidth]{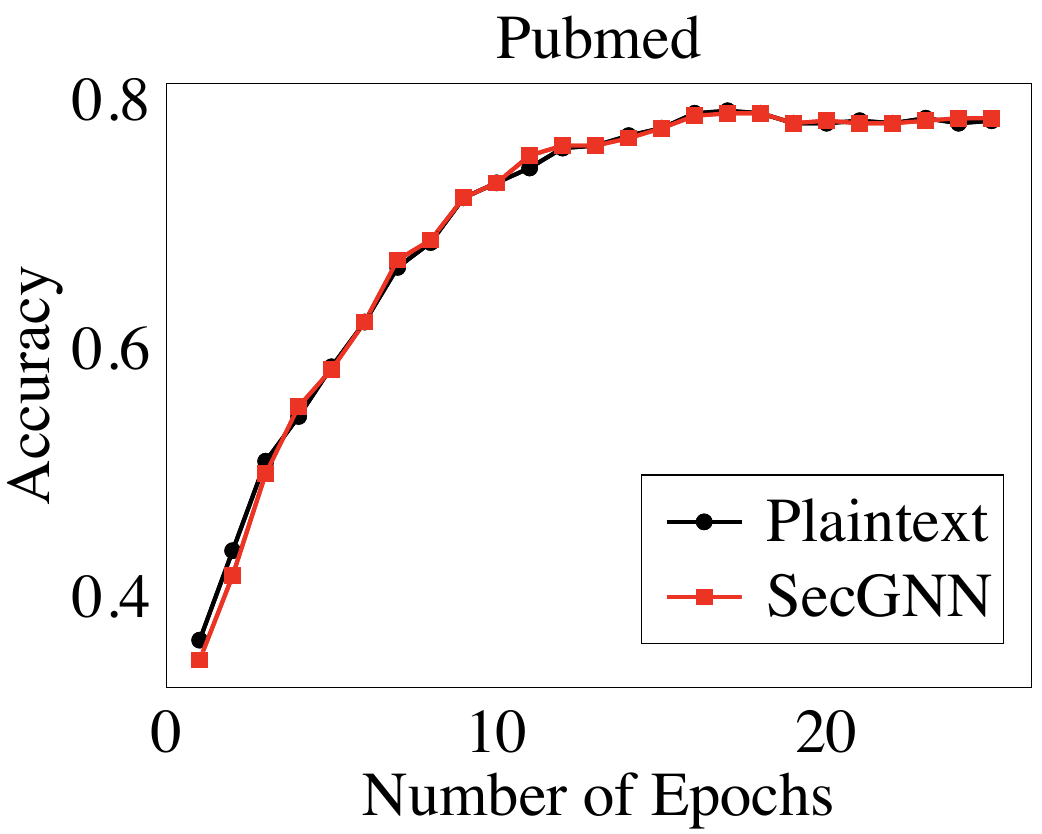}}
	\end{minipage}
	\caption{Evolution of the validation set accuracy in SecGNN and plaintext, with varying number of epochs over different datasets.}
	\label{fig:Valida}
	\vspace{-10pt}
\end{figure}

\begin{table*}
	\normalsize 
	\centering
	\caption{SecGNN's Computation and Communication Performance for Secure GNN Training and Inference}
	\label{Tab:tra_infer_overhead}
	\begin{tabular*}{\hsize}{@{}@{\extracolsep{\fill}}ccccccccc@{}}
		\toprule
		Dataset&\multicolumn{4}{c}{Training}&\multicolumn{4}{c}{Inference (a single unlabeled node)}\\\hline
		
		&   \multicolumn{2}{c}{Time (seconds)}&\multicolumn{2}{c}{Comm. (GB)}&   \multicolumn{2}{c}{Time (seconds)}&\multicolumn{2}{c}{Comm. (GB)}\\ 
		&  Online &Offline& Online &Offline&  Online &Offline& Online &Offline\\ \hline
		Citeseer&  5,640 &168& 9.1 &13.5&  49.7 &13.1& 0.7 &1\\ \hline
		
		Cora&  2,664 &54& 3.6 &5.3&  25.3&6.7& 0.4 &0.6\\ \hline
		
		Pubmed&  1,872 &72&5.1 &7.3& 69.6 &18.6 & 1 &1.5\\ \hline

	\end{tabular*}
\vspace{-10pt}
\end{table*}

\noindent\textbf{Computation and communication performance.} We now report {\secgnn}'s computation and communication performance in secure training. The results are given in Table \ref{Tab:tra_infer_overhead}, where the number of training epochs on the three datasets is as follows: Citeseer: 30, Cora: 30, and Pubmed: 25 (as shown in Fig. \ref{fig:entropy}). Over the three tested datasets, the online communication traffic in SecGNN ranges from 3.6 GB to 9.1 GB, and the online end-to-end training time varies from 31.2 minutes to 94 minutes. It is noted that the secure training procedure in SecGNN is full conducted on the cloud and the cost is one-off.
\subsection{Evaluation on Secure GNN Inference}
	\begin{table}
		\normalsize 
		\centering
		\caption{Inference Accuracy Performance}
		\label{Tab:accuracy_inference}
		\begin{tabular*}{\hsize}{@{}@{\extracolsep{\fill}}cccc@{}}
			\toprule
			Dataset& & Accuracy &Average relative error\\ \hline
			
			\multirow{2}{*}{Citeseer} &
			{\secgnn} & 68.3\% &\multirow{2}{*}{0.12\%}\\ 
			\cline{2-3}
			&Plaintext &68.3\%\\ \hline
			
			\multirow{2}{*}{Cora} &
			{\secgnn} & 78\%  &\multirow{2}{*}{0.11\%}  \\ 
			\cline{2-3}
			&Plaintext & 78\%\\ \hline
			
			\multirow{2}{*}{Pubmed} &
			{\secgnn} & 78.6\%&\multirow{2}{*}{0.12\%} \\ 
			\cline{2-3}
			&Plaintext & 78.6\%  \\ \hline
		\end{tabular*}
	\vspace{-10pt}
	\end{table}
	
\noindent\textbf{Inference accuracy.} We evaluate the Top-1 inference accuracy in {\secgnn} which performs inference with models trained in the ciphertext domain via our protocols, and compare it against with plaintext inference which is based on models trained over plaintext graphs. 
In addition, we compare the average relative error in inference results between {\secgnn} and plaintext.
Table \ref{Tab:accuracy_inference} summarizes the results, from which we can observe that the Top-1 accuracy of {\secgnn} \textit{exactly} matches that of plaintext.

\begin{figure}[t!]
	\begin{minipage}[t]{0.5\linewidth}
		\centering{\includegraphics[width=\linewidth]{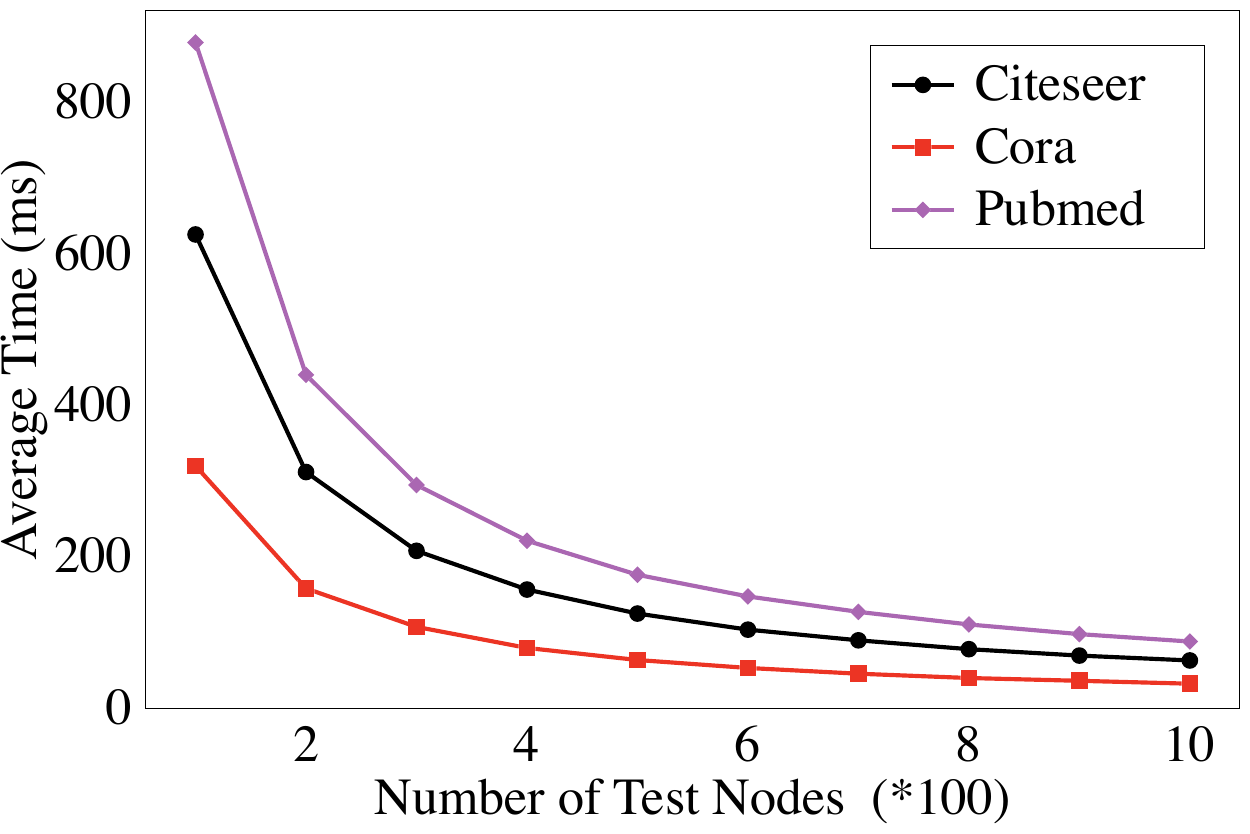}}
	\caption{Amortized runtime cost of secure inference as we vary the number of test nodes.}
		\label{fig:infer_time}
	\end{minipage}
	\begin{minipage}[t]{0.5\linewidth}
		\centering{\includegraphics[width=\linewidth]{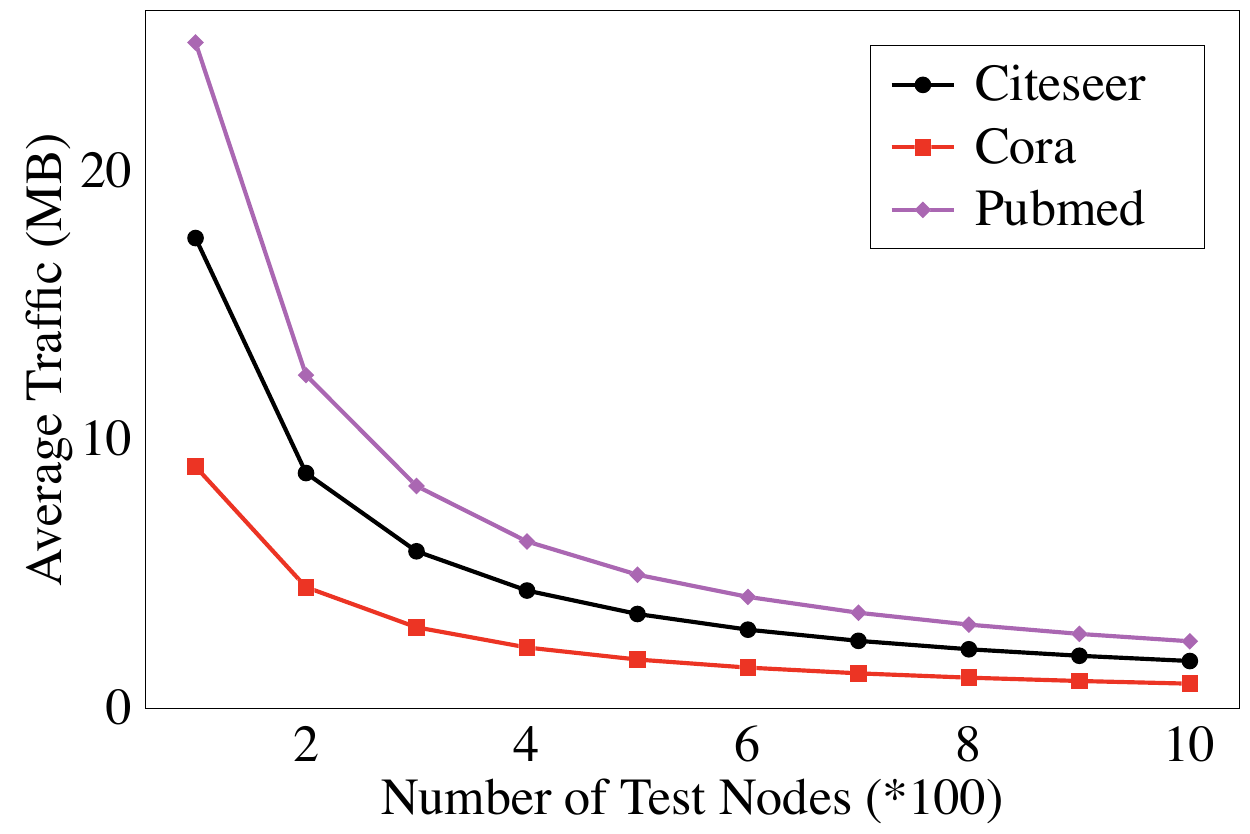}}
		\caption{Amortized traffic of secure inference, as we vary the number of test nodes.}
		\label{fig:infer_comm}
	\end{minipage}
\vspace{-10pt}
\end{figure}

\noindent\textbf{Computation and communication performance.} We examine the computation and communication performance of secure inference in SecGNN. 
Table \ref{Tab:tra_infer_overhead} shows the cost of inference for a single unlabeled node. Over the three tested datasets, the online end-to-end runtime of the sophisticated secure GNN inference for a single unlabeled node in SecGNN varies from 25.3 seconds to 69.6 seconds, with the online communication traffic ranging from 0.4 GB to 1 GB.

It is worth noting that the average cost of inferring a node's label decreases as the number of test nodes increases. That is because in secure inference, to calculate the encrypted $1_{st}$-layer aggregate state (i.e., Eq. \ref{eq:aggre1}) for a single unlabeled node, the cloud servers must calculate the encrypted $1_{st}$-layer state for all nodes since the cloud servers do not hold the IDs of the unlabeled node's neighboring nodes in plaintext. 
Therefore, if the cloud servers infer labels for a number of unlabeled nodes in a single batch, the cost can be amortized, so the average cost of individual node inference will go down. Fig. \ref{fig:infer_time} and Fig. \ref{fig:infer_comm} show the average time and communication cost with varying number of test nodes for inference.

\vspace{-10pt}
\subsection{Performance Benchmarks on Sub-Protocols}

In this section, we will first demonstrate the performance advantage of our proposed secure array access protocol over the state-of-the-art \cite{blanton2020improved} (referred to as the BYK20 protocol hereafter). After that, we evaluate the performance of secure MSB extraction which is used in secure activation functions.

\begin{table}[t]
	\normalsize 
	\centering
	\caption{Performance Comparison of Secure Array Access}
	\label{Tab:comp_array}
	\begin{tabular*}{\hsize}{@{}@{\extracolsep{\fill}}cccc@{}}
		\toprule
		& & Time (seconds)&Comm. (GB)\\ \midrule
		
		\multirow{2}{*}{Citeseer} &
		BYK20 \cite{blanton2020improved} & 22.4& 0.37\\ 
		\cline{2-4}
		&Ours& \textbf{17.9} & \textbf{0.19} \\ \hline
		
		\multirow{2}{*}{Cora} &
		BYK20 \cite{blanton2020improved}  & 9.8&0.12 \\ 
		\cline{2-4}
		&Ours &\textbf{8.1} & \textbf{0.06} \\ \hline
		
		\multirow{2}{*}{Pubmed} &
		BYK20 \cite{blanton2020improved}  &27& 0.29\\ 
		\cline{2-4}
		&Ours&  \textbf{23}& \textbf{0.15} \\ \hline
	\end{tabular*}
\vspace{-10pt}
\end{table}

\begin{figure}[t!]
	\centering
	\includegraphics[width=0.5\linewidth]{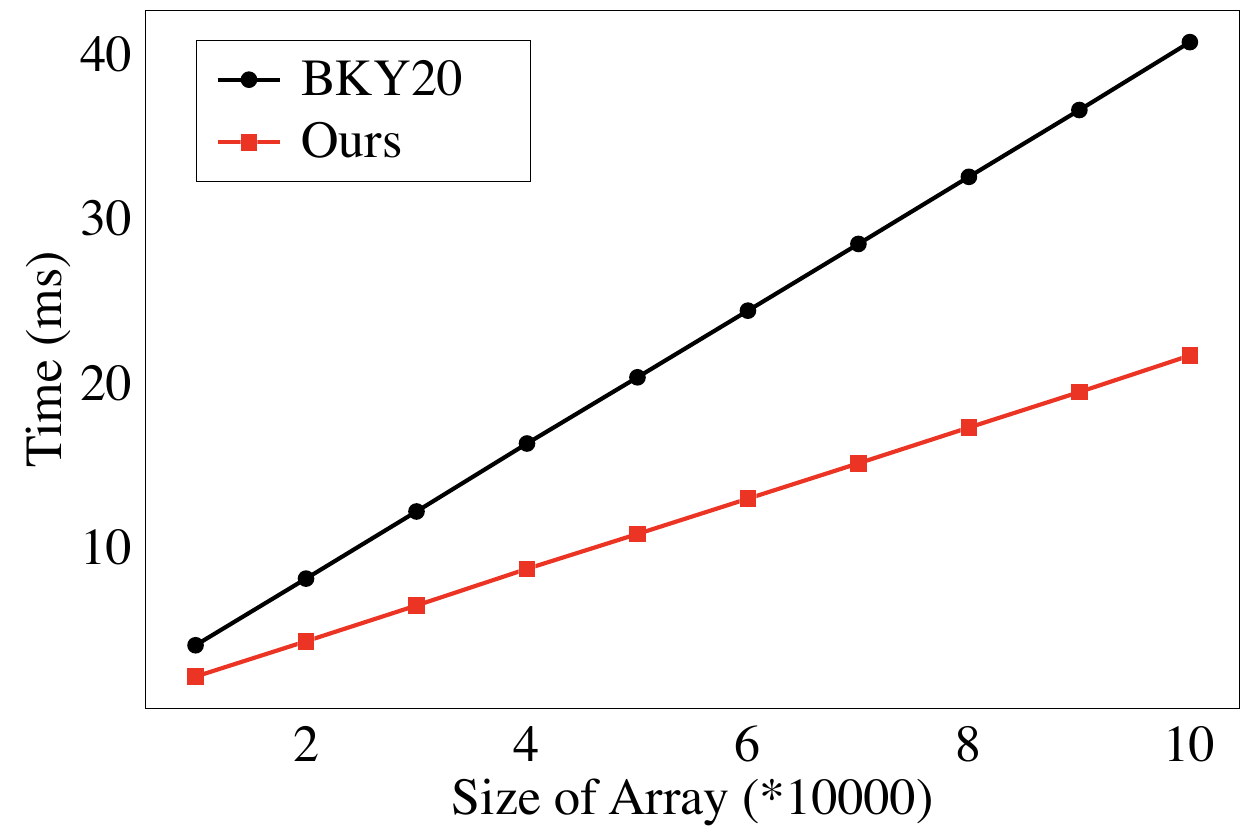}
	\caption{Runtime comparison of secure array access.}
\label{fig:array_access}
\vspace{-10pt}
\end{figure}

\noindent\textbf{Secure array access.} To demonstrate the performance advantage of our secure array access protocol over the BYK20 protocol, we evaluate the cost of securely accessing a node's feature vector from an encrypted array, where each array element is a graph node's feature vector. The size of the encrypted array is $N\times L$,
where $N$ is the number of graph nodes and $L$ is the length of each node's feature vector. 
The runtime costs are provided in Table \ref{Tab:comp_array}. 
In addition, we further compare the runtime costs of our protocol and the BYK20 protocol, with varying array sizes.
The results are given plotted in Fig. \ref{fig:array_access}.
It is observed that the efficiency gain of our protocol over the BYK20 protocol increases as the array size grows.

\begin{table}[t]
	\normalsize 
	\centering
	\caption{Theoretical Communication Performance of Secure MSB Extraction ($k=64$)}
	\label{Tab:msb64}
	\begin{tabular*}{\hsize}{@{}@{\extracolsep{\fill}}cccc@{}}
		\toprule
		Scheme &Rounds&Online (bit)&Offline (bit) \\\midrule
		$\mathrm{ABY}^{3}$ \cite{mohassel2018aby3}& 7&\textbf{677}&\textbf{0}\\ \hline
		Ours &\textbf{6}&732&1026\\
		\bottomrule
	\end{tabular*}
\vspace{-10pt}
\end{table}

\begin{table}[t]
	\normalsize 
	\centering
	\caption{Runtime Comparison of Secure MSB Extraction (in ms)}
	\label{Tab:msbexp}
	\begin{tabular*}{\hsize}{@{}@{\extracolsep{\fill}}cccc@{}}
		\toprule
		$h\times k$ &$16\times 64$&$128\times 128$&$20\times 576$ \\\midrule
		$\mathrm{ABY}^{3}$ \cite{mohassel2018aby3}& 6.08&\textbf{74.33}&\textbf{52.72}\\ \hline
		Ours &\textbf{5.89}&74.55&52.81\\
		\bottomrule
	\end{tabular*}
	\vspace{-10pt}
\end{table}

\noindent\textbf{Secure MSB extraction.} We evaluate and compare the performance of secure MSB extraction with $\mathrm{ABY}^{3}$ \cite{mohassel2018aby3}.
Table \ref{Tab:msb64} gives a comparison on the theoretical communication complexity of secure MSB extraction between SecGNN and $\mathrm{ABY}^{3}$.
The protocol in SecGNN consumes one less round, at the cost of more communication bits.
Furthermore, we conduct experiments to compare the practical efficiency under different $h\times k$ settings: number of values $\times$ bit-length.
The results are given in Table \ref{Tab:msbexp}.
It is observed that our protocol has comparable performance to $\mathrm{ABY}^{3}$ \cite{mohassel2018aby3}, and is a bit more efficient with a small size setting ($16\times 64$). 
However, it is noted that different from our system, $\mathrm{ABY}^{3}$'s security design uses replicated secret sharing, which runs among three cloud servers and needs them to interact with each other throughout the process. In contrast, we provide an alternative design to evaluate the MSB under additive secret sharing, which requires only two cloud servers $P_{\{1,2\}}$ to interact online, while the third cloud server $P_3$ just provides necessary triples in offline phase.

\section{Conclusion}
\label{sec:conclusion}

In this paper, we design, implement, and evaluate {\secgnn}, the first system supporting privacy-preserving GNN training and inference as a cloud service. Building on lightweight cryptographic techniques and a multi-server decentralized-trust setting, {\secgnn} can effectively allow the cloud servers to train a GNN model without seeing the graph data as well as provide secure inference service once the encrypted GNN model is trained. Extensive experiments on real-world datasets demonstrate that SecGNN achieves comparable plaintext training as well as inference accuracy, with practically affordable performance on the cloud. For future work, it would be interesting to explore how to extend our initial research effort to support secure GNN training and inference under a stronger active adversary model, as well as the possibility of leveraging the recent advances in trusted hardware for performance speedup.

\section*{Acknowledgement}

This work was supported in part by the Guangdong Basic and Applied Basic Research Foundation under Grant No. 2021A1515110027, and in part by the Shenzhen Science and Technology Program under Grants No. RCBS20210609103056041 and No. JCYJ20220531095416037.

\bibliographystyle{IEEEtran}
\bibliography{ref}

\end{document}